\documentclass[11pt]{article}
\usepackage{fullpage}

\usepackage{hyperref}
\usepackage[utf8]{inputenc} %
\usepackage[T1]{fontenc}    %
\usepackage{url}            %
\usepackage{mathtools}
\usepackage{booktabs}       %
\usepackage{amsfonts}       %
\usepackage{nicefrac}       %
\usepackage{microtype}      %
\usepackage{xcolor}         %
\usepackage{MnSymbol}
\usepackage{xspace}

\usepackage{tikz}
\usepackage{pgfplots}
\usetikzlibrary{calc}
\pgfplotsset{compat=1.18}

\usepackage{pict2e,picture,graphicx}
\usepackage[overload]{empheq}
\usepackage{microtype}
\usepackage{graphicx}
\usepackage{subcaption}
\usepackage{booktabs}
\usepackage{amsmath}
\usepackage{cases}
\usepackage{mathtools}
\usepackage{amsthm}
\usepackage{thm-restate}
\usepackage{enumitem}

\allowdisplaybreaks
\usepackage{xparse}
\usepackage[capitalize, noabbrev]{cleveref}
\usepackage{wrapfig}
\usepackage{bbm}
\usepackage{yfonts}
\usepackage{nicefrac} 
\usepackage{multirow}
\usepackage{bm}
\usepackage{bbm}
\usepackage{subcaption}

\hypersetup{colorlinks,
    colorlinks = true,
    linkcolor=DarkBlue,
    citecolor=green!60!black,
    urlcolor=green!70!black,
    filecolor=DarkBlue
    linktocpage = true,
}
\usepackage{comment}

\newcommand{\E}{\mathop{\mathbb{E}}}

\colorlet{darkgreen}{green!70!black}

\newcommand{\CFont}[1]{{\textup{\textsf{#1}}}\xspace}

\newcommand{\NP}{\CFont{NP}}

\newcommand{\AM}{\CFont{AM}}
\newcommand{\FG}{\CFont{FreeGame}}

\newcommand{\OPT}{\textnormal{OPT}}
\newcommand{\Sat}{\CFont{Satisfiability}}

\newcommand{\cD}{\mathcal{D}}

\newcommand{\poly}{\textsf{\textup{poly}}}
\newcommand{\polylog}{\textnormal{polylog}}

\newcommand{\Naturals}{\mathbb{N}}

\usepackage{color}
\definecolor{mygreen}{rgb}{0.0, 0.5, 0.0}
\definecolor{myorange}{rgb}{0.55, 0.62, 1}
\newcommand{\nb}[3]{}

\theoremstyle{plain}

\makeatletter
\let\cref@old@stepcounter\stepcounter
\def\stepcounter#1{%
  \cref@old@stepcounter{#1}%
  \cref@constructprefix{#1}{\cref@result}%
  \@ifundefined{cref@#1@alias}%
    {\def\@tempa{#1}}%
    {\def\@tempa{\csname cref@#1@alias\endcsname}}%
  \protected@edef\cref@currentlabel{%
    [\@tempa][\arabic{#1}][\cref@result]%
    \csname p@#1\endcsname\csname the#1\endcsname}}
\makeatother

\theoremstyle{plain}
\newtheorem{theorem}{Theorem}[section]

\newtheorem{claim}[theorem]{Claim}
\newtheorem{lemma}[theorem]{Lemma}

\newtheorem*{theorem-non}{Theorem}

\theoremstyle{definition}
\newtheorem{definition}[theorem]{Definition}

\theoremstyle{remark}

	\definecolor{niceRed}{HTML}{BE2626}
	\definecolor{Red2}{HTML}{DB3236}
	\definecolor{mgreen}{HTML}{A0C88C}
	\definecolor{blueGrotto}{HTML}{059DC0}
	\definecolor{limeGreen}{HTML}{81B622}
	\definecolor{myellow}{HTML}{E09B23}
	\definecolor{darkGreen}{HTML}{2E8B57}
	\definecolor{navyBlueP}{HTML}{03468F}
	\definecolor{Sepia}{HTML}{7F462C}
	\definecolor{orange2}{HTML}{FF8000}
	\definecolor{mgray}{HTML}{ABB3B8}
	\definecolor{lgray}{HTML}{E5E8E9}
	\definecolor{myPurple}{HTML}{AF007C}
	\definecolor{mypurple2}{HTML}{CC9EFF}
	\definecolor{royalBlue}{HTML}{057DCD}
	\definecolor{mpink}{HTML}{FC6C85}
	\definecolor{lblue}{HTML}{4A90E2}
	\definecolor{peagreen}{HTML}{98C127}
	\definecolor{typnavy}{HTML}{001F3F}
	\definecolor{typblue}{HTML}{0074D9}
	\definecolor{typaqua}{HTML}{7FDBFF}
	\definecolor{typteal}{HTML}{39CCCC}
	\definecolor{typeastern}{HTML}{239DAD}
	\definecolor{typpurple}{HTML}{B10DC9}
	\definecolor{typfuchsia}{HTML}{F012BE}
	\definecolor{typmaroon}{HTML}{85144B}
	\definecolor{typred}{HTML}{FF4136}
	\definecolor{typorange}{HTML}{FF851B}
	\definecolor{typyellow}{HTML}{FFDC00}
	\definecolor{typolive}{HTML}{3D9970}
	\definecolor{typgreen}{HTML}{2ECC40}
	\definecolor{typlime}{HTML}{01FF70}
	\definecolor{newgreen}{HTML}{83C702}
	\definecolor{mathchaYellow}{HTML}{F8E71C}
	\definecolor{mathchaGreen}{HTML}{7ED321}
	\definecolor{mathchaPurple}{HTML}{BD10E0}
	\definecolor{mathchaBlue}{HTML}{58B1FF}
	\definecolor{mathchaRed}{HTML}{FF5B59}
	\definecolor{metablue}{HTML}{0064E0}
    \definecolor[named]{Purple}{cmyk}{0.55,1,0,0.15}
    \definecolor[named]{DarkBlue}{cmyk}{1,0.58,0,0.21}

\usepackage[style=alphabetic,natbib=true,maxcitenames=2, maxbibnames=10,backend=bibtex]{biblatex}
\allowdisplaybreaks

\title{A Subsampling Theorem for Constraint Satisfaction Problems with Large Arity} %
\author{
 Martino Bernasconi \\
 Bocconi University\\
 {\textcolor{black}{\small\texttt{martino.bernasconi@unibocconi.it}}}
 \and 
 Matteo Castiglioni  \\
 Politecnico di Milano\\
 {\textcolor{black}
 {\small\texttt{matteo.castiglioni@polimi.it}}}
 \and 
 Andrea Celli  \\
 Bocconi University\\
 {\textcolor{black}
 {\small\texttt{andrea.celli2@unibocconi.it}}}
 \and 
 Gabriele Farina  \\
 MIT\\
 {\textcolor{black}
 {\small\texttt{gfarina@mit.edu}}}
 \and 
 Giulio Malavolta  \\
 Bocconi University\\
 {\textcolor{black}
 {\small\texttt{giulio.malavolta@unibocconi.it}}}
}
\newcommand{\val}{\textnormal{val}}
\newcommand{\bal}{\textnormal{bal}}
\newcommand{\cM}{\mathcal{M}}

\newcommand{\cW}{\mathcal{W}}
\newcommand{\ccS}{\mathcal{S}}

\date{}

\begin{document}

\maketitle
\begingroup
\renewcommand{\thefootnote}{}
\footnotetext{\hspace{-6mm}\textbf{Acknowledgments.} Martino Bernasconi and Andrea Celli were supported by an ERC grant (Project 101165466 — PLA-STEER). Matteo Castiglioni was supported by the EU Horizon project ELIAS (European Lighthouse of AI for Sustainability, No. 101120237). Gabriele Farina was supported in part by the National Science Foundation award CCF-2443068, the Office of Naval Research grant N000142512296, and an AI2050 Early Career Fellowship. Giulio Malavolta is supported by the European Research Council through an ERC Starting Grant (Grant agreement No.~101077455, ObfusQation) and by the Deutsche Forschungsgemeinschaft (DFG, German Research Foundation) under Germany's Excellence Strategy - EXC 2092 CASA – 390781972.

The authors thank Scott Aaronson, David Steurer, and Gal Arnon for helpful comments and discussions.}
\addtocounter{footnote}{-0}
\endgroup

\begin{abstract}
    Subsampling theorems for constraint satisfaction problems (CSPs) guarantee that the value of the CSP is approximately preserved after restricting it to small random subsets of variables.
    We provide the first subsampling theorem for CSPs, which requires a sample size that is polynomial in the arity $k$ and error $\varepsilon$, and polylogarithmic in the alphabet size $q$. This improves upon the subsampling theorem of Barak, Hardt, Holenstein, and Steurer~(SODA '11), which achieves a polynomial dependency on $\varepsilon$ and polylogarithmic in $q$ only in the constant-arity regime.
    Our subsampling theorem has applications in interactive proofs and property testing. In interactive proofs, it provides a key missing ingredient for the proof of Aaronson,  Impagliazzo, and Moshkovitz~(CCC '14) that $\textsf{AM}(\textsf{poly})=\textsf{AM}$ (where $\textsf{AM}(k)$ is the class of languages decidable by Arthur-Merlin protocols with $k$ non-communicating Merlins with independent questions). In property testing, it yields the first one-sided tester for satisfiability with sample size polynomial in the arity $k$ and the error $\varepsilon^{-1}$, and polylogarithmic in the alphabet size $q$.
\end{abstract}

\pagenumbering{gobble} 

\newpage
\tableofcontents
\newpage
\pagenumbering{arabic}

\section{Introduction}

Constraint satisfaction problems (CSPs) occupy a central place in approximation algorithms, property testing, and interactive proof complexity. In a $k$-CSP, one is given predicates on $k$-tuples of variables over an alphabet $[q]$, and the goal is to find an assignment
that maximizes the total value of the constraints. 
When the instance is \emph{dense}---that is, when a constant fraction of all $k$-tuples carry constraints---it exhibits algorithmic phenomena that are absent in sparse regimes. In particular,
dense CSPs admit approximation schemes and sublinear-time algorithms that are believed to be impossible in sparse regimes.

There has been a longstanding effort towards approximating dense CSPs. The systematic study of dense instances of \NP-hard problems was initiated by
\citet*{arora1995polynomial}, who gave a $O(2^{\poly(1/\varepsilon,2^k)})$ polynomial-time approximation scheme for Boolean CSPs (i.e., with $q=2$). 
Subsequent work developed a variety of techniques for this regime, including exhaustive sampling \citep{arora1995polynomial,de2005tensor,mathieu2008yet,fotakis2015sub},
semidefinite programming hierarchies \citep{raghavendra2008optimal,guruswami2011lasserre,goemans1995improved,khot2007optimal,de2007linear,yoshida2014approximation}, and graph-regularity lemmas \citep{arora1995polynomial, frieze1996regularity, de2005tensor}.

An especially powerful and conceptually simple approach is \emph{subsampling}: randomly sampling a small subset of variables induces a subinstance whose optimum is, in expectation, close to that of the original instance.
This phenomenon was first observed in graph problems by
\citet*{goldreich1998property}, and later extended to CSPs by \citet*{alon2002random}, who showed that for every fixed arity $k$, the optimum of a dense Boolean $k$-CSP can be approximated from a random subinstance whose size depends only on the desired accuracy parameter $\varepsilon$.
\citet*{barak2011subsampling} extended these results to non-Boolean CSPs and more general notions of density, showing that, for constant $k$, a random subset of size $\poly( \varepsilon^{-1}\log q)$ suffices to obtain an $\varepsilon$-additive approximation.\footnote{The dependence on $k$ is not explicit in the theorem statement, but it becomes apparent in the proof that the dependence is exponential in $k$. This was also confirmed by David Steurer in private correspondence.}

Despite this progress, subsampling theory remains largely confined to the fixed-arity regime. To date, the strongest known subsampling theorems require a sample size that grows exponentially in $k$. Obtaining subsampling guarantees with polynomial dependence on $k$ remains an open problem---both as a natural question in its own right and essential for the applications discussed next.

Our main result is the first subsampling theorem with polynomial dependence on the arity $k$ and polylogarithmic dependence on the alphabet size $q$.

\begin{theorem-non}[Informal version of \Cref{thm:subsampling}]
Consider a dense $k$-CSP over alphabet $[q]$ with optimal value $\mathrm{OPT}$. For any $\varepsilon>0$, let $U$ be a uniformly
random subset of variables of size
\[
N= \tilde\Theta\left(\frac{k^{12}}{\varepsilon^6}\log^2\left( q\right)\right),
\]
and let $\OPT(U)$ be the optimal value of the subinstance restricted to $U$. Then
\[
|\mathbb{E}_U\left[\OPT(U)\right]-\OPT| \le \varepsilon.
\] 
\end{theorem-non}
This theorem has several immediate applications. Here, we focus on its implications for Arthur-Merlin games and property testing. 
In particular, it supplies a key ingredient in the argument of \citet*{aaronson2014multiple} establishing that $\AM(\poly)=\AM$.
While \citet{aaronson2014multiple} cite \citet{barak2011subsampling} for such a result, the latter does not provide the correct dependence on $k$ (see \Cref{sec:AMintro} for details). Our stronger subsampling theorem therefore provides the missing ingredient needed by \citet{aaronson2014multiple} to complete the proof of $\AM(\poly)=\AM$.
As an application to property testing \citep{goldreich2017introduction}, we obtain the first $\poly(k/\varepsilon)\polylog(q)$ one-sided error tester for satisfiability problems with arity $k$, and alphabet size $q$ (see \Cref{sec:ptintro} for details).

\subsection{Arthur-Merlin Games}\label{sec:AMintro}

\citet{aaronson2014multiple} introduced the class $\AM(k)$ as the set of languages decided by $k$ non-communicating provers (Merlins) interacting with a polynomial-time verifier (Arthur), under the crucial restriction that the verifier samples one question for each prover independently and uniformly.

\begin{definition}[$\AM(k)$]
    For an integer $k\in\mathbb{N}$, $\AM(k)$ is the class of languages $L\subseteq\{0,1\}^*$ for which there exist a polynomial-time verifier $V$
    such that for all $n$, there exist finite question sets $(Y_i)_{i\in[k]}$, and finite answer sets $(B_i)_{i\in[k]}$, all of size at most $2^{\poly(n)}$, %
    such that for all $x\in\{0,1\}^n$:
    \begin{itemize}[leftmargin=0.4cm]
    \item if $x\in L$, there exists $(b_i:Y_i\to B_i)_{i\in[k]}$ such that $\mathbb{P}_{y_i\sim Y_i}[V(x,y_1,\ldots,y_k,b_1(y_1),\ldots, b_k(y_k))]\ge 2/3$,
    \item if $x\notin L$, then for every $(b_i:Y_i\to B_i)_{i\in[k]}$ we have $\mathbb{P}_{y_i\sim Y_i}[V(x,y_1,\ldots,y_k,b_1(y_1),\ldots, b_k(y_k))]\le 1/3$.
    \end{itemize}
\end{definition}

Related to the class $\AM(k)$, we define the following optimization problem faced by the $k$ Merlins, which we call \FG.

\begin{definition}[\FG]\label{def:fg}
    A \FG instance $G=((Y_i)_{i\in[k]},(B_i)_{i\in[k]},V)$ is defined by finite question sets $(Y_i)_{i\in[k]}$, finite answer sets $(B_i)_{i\in[k]}$, and a function $V:Y_1\times\cdots\times Y_k\times B_1\times\cdots\times B_k\to[0,1]$. The value of the game is
    \[
    \omega(G):=\max_{(b_i:Y_i\to B_i)_{i\in[k]}} \mathbb{E}_{y_i\sim Y_i}[V(y_1,\ldots, y_k,b_1(y_1),\ldots, b_k(y_k))].
    \]
\end{definition}

Allowing the number of Merlins to scale polynomially in the size of the input yields the class $\AM(\poly)=\bigcup_{c\in\mathbb{N}}\AM(n^c)$.

One of the main results of \citet{aaronson2014multiple} is a reduction from $\mathsf{AM}(k)$ to $\mathsf{AM}$ based on subsampling.
The central insight of \citet{aaronson2014multiple} is that \FG can be reduced to dense CSPs: the questions correspond to variables, and the provers' strategies define assignments. Under this correspondence, the value $\omega(G)$ coincides with the optimum of the resulting $k$-CSP.

Applying subsampling to the resulting CSP reduces the total number of constraints and allows one to simulate $k$ Merlins with a single prover as follows: Arthur sends all subsampled questions to the single Merlin, who explicitly returns the strategy of each of the $k$ Merlins on those questions. Arthur can then locally estimate the expected value on the subsampled constraints.

For Arthur to run in polynomial time, the number of sampled variables must itself remain polynomial. When $k$ grows, i.e., $k=\poly(n)$, this requires a subsampling theorem with a polynomial dependence in the arity $k$. 
\citet{aaronson2014multiple} cite \citet{barak2011subsampling} for such a theorem, but a closer inspection of the proof reveals that \citet{barak2011subsampling} treats $k$ as a constant, and in fact the number of samples grows exponentially with $k$. %

Therefore, our subsampling theorem, whose sample complexity is polynomial in the arity $k$ and in the inverse approximation error $\varepsilon^{-1}$, and polylogarithmic in the alphabet $q$, supplies the missing ingredient needed to complete the proof of \citet{aaronson2014multiple}. This confirms that $\mathsf{AM}(\poly) = \mathsf{AM}$.
\begin{theorem-non}[Corollary of \Cref{thm:subsampling-k-free-games}]
    $\AM(\poly)=\AM$.
\end{theorem-non}

\subsection{Property Testing}\label{sec:ptintro}

Property testing \citep{rubinfeld1996robust, goldreich1998property, goldreich2017introduction} studies the problem of designing randomized algorithms that, given query access to an object (e.g., a graph), distinguish between the case where the object satisfies a given property and the case where it is far from satisfying it, using only a small number of sampled variables.

A standard problem in this area is testing the satisfiability of a CSP with Boolean-valued constraints.

\begin{definition}[$(k,q)$-\Sat]
A Boolean-valued $k$-CSP on $n$ variables over alphabet $[q]$:
\begin{itemize}
\item is \emph{satisfiable} if there exists an assignment that satisfies all constraints;
\item is \emph{$\varepsilon$-far} from satisfiable if at least $\varepsilon n^k$ constraints must be removed to make it satisfiable.
\end{itemize}
\end{definition}

The one-sided satisfiability testing problem requires designing a randomized testing algorithm that:
(i) always accepts if the CSP is satisfiable, and
(ii) rejects instances that are $\varepsilon$-far from satisfiable with probability at least $2/3$. 
On the other hand, a two-sided tester has probabilistic guarantees on both sides.

In property testing, algorithms are evaluated by their sample complexity, which, in the context of CSPs, is the number of variables sampled. Property testing for CSPs is closely related to subsampling theorems, and our subsampling theorem can be easily translated into a tester with the following sample complexity.

\begin{theorem-non}[Informal version of \Cref{th:proptestingSAT}]
    There exists a one-sided $\varepsilon$-tester for $(k,q)$-\Sat with sample complexity $\tilde O({k^{12}}\log^2 \left( q\right)/{\varepsilon^{7}})$.
\end{theorem-non}

The study of $(k,q)$-\Sat was initiated by \citet{alon2003testing}, who provide a tester with sample complexity $2^{O(q^{2k})}/\varepsilon^2$. This bound was later improved to $O(q^{3k}/\varepsilon)$ by \citet{sohler2012almost}.
While it is known that any tester requires at least $\Omega(1/\varepsilon)$ samples \citep{alon2002testing}, to the best of our knowledge, no lower bounds are known in terms of the alphabet size $q$. 

Recently, \citet{blais2024new} provided a one-sided tester with sample complexity $\tilde O(qk^3/\varepsilon)$, based on the hypergraph container method introduced to property testing by \citet{blais2025testing}. This was the first bound polynomial in $1/\varepsilon,k$ and $q$. 
Our result provides an exponential improvement in the dependency on the alphabet size $q$, while maintaining a polynomial dependency on both $1/\varepsilon$ and $k$. Our result does make progress on the question of whether a $\poly(\log q)\tilde O(k/\varepsilon)$ tester exists for $(k,q)$-\Sat \citep{blais2024new,gishboliner2025polynomial}. However, even if our analysis never tried to optimize constants, obtaining linear dependence on $k$ and $1/\varepsilon$ from our techniques seems difficult, due to our nested concentration technique (see \Cref{sec:technicaloverview}). 

We also remark that, for the case of fixed $k$, an upper bound polylogarithmic in $q$ was implicitly provided by the $\poly(\varepsilon^{-1}\log q)$ subsampling theorem of \citet{barak2011subsampling}.  
Indeed, from the weaker subsampling theorem of \citet{barak2011subsampling}, one could obtain a tester for $(O(1),q)$-\Sat with sample complexity $\poly(1/\varepsilon,\log q)$ when $k=O(1)$.

\subsection{Technical Overview}\label{sec:technicaloverview}

The proof of our subsampling theorem follows the double-sampling idea of \citet*{goldreich1998property} and the high-level blueprint of \citet*{barak2011subsampling}. The argument is divided into two main steps.

\begin{enumerate}[label=(\roman*)]
\item \textbf{Concentration lemma:} Given a CSP and a fixed assignment $\phi$, consider a random subset of variables $U$. The concentration lemma shows that, with high probability, the value of $\phi$ on the induced subinstance $U$ is close to its value on the original instance. 

\item \textbf{Structure lemma:} The goal of the structure lemma is to show that, for a typical subsample $U$, the optimum on the subinstance is well approximated by the best assignment in a small fixed family $\Psi$. The key challenge is that $\Psi$ must be constructed without knowledge of $U$ and yet it must cover the optimum for all typical $U$ simultaneously. Crucially, the existence of a family $\Psi$ is shown by a probabilistic argument that relies on the concentration lemma above.
\end{enumerate}

Our main innovation is a sharper concentration analysis (\Cref{lem:concentration}) with only polynomial dependence on the arity $k$. This tool is used twice in the analysis: once in combination with the structure lemma to obtain the final subsampling theorem, and once \emph{within} the proof of the structure lemma itself. Thus, although our structure lemma follows the double-sampling architecture introduced by \citet{goldreich1998property} as extended by \citet{barak2011subsampling}, its analysis requires several refinements in order to operate in the growing-arity regime. We discuss the required improvements below.

\newcommand{\xparagraph}[1]{\vspace{3mm}\noindent\textbf{#1}~}

\xparagraph{Improved Concentration Analysis.}
We begin discussing our concentration result. Our analysis is based on the algebraic representation of the CSP objective as a multilinear polynomial over the simplex. This differs from the concentration analysis of \citet{barak2011subsampling}, which instead follows an iterative pruning argument to control variable influences, which introduces error compounding across the arity $k$. In contrast, our algebrization is not iterative, helping to avoid such overhead. 
More specifically, we represent an assignment with a \emph{single} $nq$-simplex in which a multilinear degree-$k$ polynomial $f$ represents the values of the assignment. Then, we sample from the joint categorical distribution on this space (inducing an empirical distribution $Z$), and show that the mean of the sampled values concentrates around the assigned value. 
The variables sampled induce the random subinstance $U$.
Due to the non-linearity of the CSP algebrization, two errors need to be kept under control: a concentration error $|f(Z)-\mathbb{E}[f(Z)]|$, and a bias term $|\mathbb{E}[f(Z)]-f(\mathbb{E}[Z])|$. The concentration error is controlled via McDiarmid's inequality, while the bias is controlled via upper bounds on the factorial moments of multinomial distributions. Crucially, both the concentration error and the bias have only polynomial dependence on $k$. This contrasts with the iterative pruning analysis of \citet{barak2011subsampling}, in which the dependence on the arity compounds across the $k$ stages of the argument. 

\xparagraph{A Refined Structure Lemma.}
As mentioned above, the improved concentration bound enters the overall proof in two places. First, for any fixed family of assignments $\Psi$, we use the concentration result to guarantee that the optimum over $\Psi$ on the subsampled instance is close to the optimum over $\Psi$ on the full instance; this follows by a union bound over $\Psi$.
Second, and more subtly, we use the concentration bound within the proof of the structure lemma (\Cref{lm:structure}) itself. %
As in the double-sampling framework of \citet{goldreich1998property}, the structure lemma constructs a small family $\Psi$ from assignments to a collection of seed variables. Once a subsample $U$ is fixed, a guessed assignment to the seeds is extended to a full assignment through an iterative best-response procedure. Computing the best response of a variable $i$ to a label $a$ amounts to evaluating a derived $(k-1)$-CSP, and the seed is used as a proxy for the much larger set of variables outside the block containing $i$. Our concentration lemma guarantees that these proxy values accurately estimate the corresponding full values.

Obtaining a seed of size $\poly(k/\varepsilon)$ is the key difficulty of this proof, and does not amount to substituting the new concentration bound into the analysis of \citet{barak2011subsampling}.
Crucially, maintaining the $\poly(k)$ dependency across the entire argument requires substantial re-engineering of the original arguments. This analysis shows that a seed of size $\poly(k/\varepsilon,\log q)$ suffices, from which we can generate a family of size $|\Psi|=q^{\poly(k/\varepsilon)}$, small enough to allow the final union bound. %
We refer to \Cref{sec:structure} for an overview of the proof.

\section{Constraint Satisfaction Problems}

We first introduce some basic mathematical notation. We let $[n]=\{1,\ldots, n\}$ for any $n\in \Naturals$, and for any $k\in\Naturals$, we define the falling factorial $n^{\underline{k}}=n\cdot(n-1)\cdots(n-k+1)$. For any $n\in \Naturals$, $\Delta_n$ is the set of discrete distributions over $[n]$.

Then, we formally define CSPs and the related problem Max-$k$-CSP.

\begin{definition}[CSP]
    A \emph{$k$-constraint satisfaction problem} ($k$-CSP) instance on $n$ variables with alphabet size $q$ is specified by a collection of constraints $P_w$. Here, $w$ is a $k$-tuple of distinct variable indices (we denote by $\cW$ the set of all such $k$-tuples), and $P_w: [q]^k \to [0,1]$ is a \emph{bounded predicate}.
    An assignment is a function $\phi \in [q]^n$. The value of an assignment is defined as:
    \[
    \val(\phi) = \frac{1}{n^k} \sum_{w \in \cW} P_w(\phi|_w),
    \]
    where $\phi|_w$ denotes the assignment restricted to the variables indexed by $w$.
\end{definition}

Notice that we normalize by $n^k$ even though $|\cW| = n^{\underline{k}}$. As a result, the true maximum value of a CSP is $n^{\underline{k}}/n^k$, which approaches $1$ for $k = o(n)$. We adopt this convention because normalizing by $n^k$, while maintaining that constraints are defined over $k$ distinct variables, greatly simplifies the proofs.

The Max-$k$-CSP problem asks to compute $\OPT = \max_{\phi \in [q]^n} \val(\phi) \in [0,1]$. Since we consider a normalized objective, our focus will be on finding $\varepsilon$-additive approximations of $\OPT$.

We also introduce notation for a CSP restricted to a subset of variables $U$. Given a subset of variables $U \subseteq [n]$, let $\cW(U) = \cW \cap U^k$ be the set of constraints supported entirely on $U$. This allows us to define the values and optimum restricted to $U$:

\[
\val_U(\phi)=\frac{1}{|U|^{{k}}}\sum_{w\in \cW(U)} P_w(\phi|_w),\quad
\textnormal{and} 
\quad
\OPT(U)=\max_{\phi:[q]^n}\val_U(\phi|_U).
\]

Moreover, we will need to combine assignments. Given an assignment $\phi \in [q]^{U_1}$ and an assignment $\psi \in [q]^{U_2}$ over disjoint sets of variables ($U_1 \cap U_2 = \emptyset$), we denote their joint assignment over $U_1 \cup U_2$ as $\phi \otimes \psi$. We also use $(i \rightarrow a)$ to denote the partial assignment which assigns $a\in [q]$ to variable $i \in [n]$.

\subsubsection*{Subsampling CSPs}

Subsampling theorems for CSPs state that a sufficiently large, uniformly random subset of variables $U\subseteq [n]$ approximates the optimal value of the CSP. More precisely,
\[
\big|\mathbb{E}_U[\OPT(U)]- \OPT\big|\le\varepsilon,
\]
where $N=|U|$, which usually is a function of just $\varepsilon,k$ and $q$, determines the accuracy of the approximation.

\section{Subsampling Theorem}

In this section, we present the statement and the high-level structure of our subsampling theorem.
Following the blueprint of \citet{barak2011subsampling}, our proof is split into two primary components: 1) a \emph{concentration lemma} for fixed assignments, and 2) a \emph{structure lemma} which guarantees the existence of a small set of assignments that approximately capture the optimum of a random subsampled instance.

The concentration lemma shows that for any fixed assignment $\phi$, a small random subset of variables $U$ provides an accurate estimate of $\val(\phi)$ by computing the average value of the constraints over variables in $U$. This is where we make a significant quantitative improvement over previous proofs. While prior work often yielded bounds on the sample size that were exponential in the arity (e.g., $O(\poly(\varepsilon^{-k}))$), we prove that a subset of size $O(\text{poly}(k/\varepsilon)\log(1/\delta))$ suffices to obtain an $O(\varepsilon)$-approximation with probability $1-\delta$.
Formally, we show that:

\begin{lemma}[Concentration Lemma]\label{lem:concentration}
    Let $\phi\in [q]^n$ be an assignment.
    Sample $N$ variables uniformly with replacement from $[n]$, and let $U$ be the resulting set of unique variables.
    For any $\varepsilon>0$, with probability at least $1-2 \exp(-\frac{\varepsilon^2 N}{2k^2})$, it holds:
    \[
    \left|\val_U(\phi)-\val(\phi)\right| \le \varepsilon+\frac{k^2}{N} + \mathbb{I}[|U|\neq N].
    \]
\end{lemma}

For technical reasons, we proved concentration for a set $U$ sampled \emph{with} replacement. While we ultimately require the result for sampling \emph{without} replacement, the two settings are essentially equivalent in our regime.
Indeed, we get that $\mathbb{I}[|U|\neq N]=0$ holds with high probability (see \Cref{clm:withToWithout}).
The reason we need this special structure is the following. We could, in principle, do a union bound with the event ``there are no collisions'' (i.e., $|U|\neq N$) when sampling with replacement (this event has probability at least $1-\Omega(N^2/n)$); however, we will need to do union bounds on a large number of events, and the term $O(N^2/n)$ is not small enough to absorb such union bounds.

The second lemma shows the existence of a small set of ``representative'' assignments $\Psi$. Although $|\Psi|$ will be small, it will be sufficiently diversified to capture the optimum on a random subsample $U$. This approach was pioneered by \citet{goldreich1998property} for {Max-CUT} and later generalized by \citet{barak2011subsampling} to CSPs. Here, we refine \citet{barak2011subsampling} ``structure lemma'' along different directions. The more evident one follows from our improved concentration lemma.
Indeed, one crucial step in the structure lemma is to apply the concentration lemma to a derived CSP of arity $k-1$. In particular, we essentially use a ``double'' sampling, introducing an additional parameter $s<N$ which appears in our structure lemma.  Here, we achieve an exponential improvement in the dependence on the arity $k$. We will discuss these aspects in more detail in \Cref{sec:structure}.

\begin{lemma}[Structure Lemma]  \label{lm:structure}
For every $\varepsilon>0$, $N\in [n]$, and $s<N$, there exists a set of assignments $\Psi\subseteq [q]^n$ with $|\Psi|=q^{k^2s/\varepsilon}$ such that for uniform random subset $U$ of size $N$, we have:
\[
\mathbb{E}[\OPT(U)]\le \mathbb{E}\left[\max_{\phi\in\Psi} \val_U(\phi)\right]+\mathcal{E}(\varepsilon,N,s).
\]
where 
\(
\mathcal{E}(\varepsilon,N,s)=O\left( k\varepsilon+\frac{k^3}{s}
    + qN\exp\left(-\frac{\varepsilon^2 s}{2k^2}\right)
    + \frac{k^2s^2}{\varepsilon N}
    + \frac{sk^3}{\varepsilon N} 
    +  \sqrt\frac{k^4}{\varepsilon N}  + \frac{k^2}{\varepsilon} \exp\left(-\frac{\varepsilon^2 N}{k^4}\right)\right).
\)
\end{lemma}

By combining \Cref{lem:concentration} and \Cref{lm:structure}, we derive our subsampling theorem. The core intuition is that the structure lemma allows us to approximate the value of the optimal assignment on $U$ by taking a union bound over the restricted family $\Psi$, whose size is $q^{\tilde{\Theta}(\text{poly}(k/\varepsilon))}$, rather than over the entire assignment space $[q]^n$.

\begin{theorem}[Subsampling] \label{thm:subsampling}
There exists a positive absolute constant $c$ such that letting $N=c\frac{k^{12}}{\varepsilon^6} \log^2(qk^2/\varepsilon)$, $U$ be a set of variables of size $N$ sampled uniformly without replacement from $[n]$, and assuming $n\ge  c^2 \frac{k^{24}}{\varepsilon^{13}}\log^4(qk^2/\varepsilon)$, it holds
        \[
        \big|\mathbb{E} [\OPT(U)]-\OPT\big|\le O(\varepsilon).
        \]
\end{theorem}

\begin{proof}

The first inequality is straightforward: for the optimal assignment $\phi^*$ of the full CSP, and directly compute
\[
\mathbb{E}[\val_U(\phi^*)] = \frac{1}{N^k}\sum_{w\in \cW}P_w(\phi^*|_w)\mathbb{P}[w\subseteq U]=\frac{N^{\underline{k}}}{N^k}\frac{n^k}{n^{\underline{k}}}\val(\phi^*).
\]
Now note that, by \Cref{lem:trivialtmp}, $\frac{N^{\underline{k}}}{N^k}\frac{n^k}{n^{\underline{k}}}\ge \frac{N^{\underline{k}}}{N^k}\ge 1-\frac {k^2}N\ge 1-\varepsilon$, and thus we obtain 
\[
\OPT-\varepsilon\le\mathbb{E}[\val_U(\phi^*)]\le \mathbb{E}[\OPT(U)].
\]

The most challenging part is proving the second inequality.
Let $\Psi$ be the set from \Cref{lm:structure}.
Sample $N$ variables uniformly with replacement from $[n]$, and let $U$ be the resulting set of unique variables.
 By \Cref{lem:concentration} and a union bound over all $\phi \in \Psi$, we get that with probability at least $1 - \delta$:
\begin{align}\label{eq:concSub}
    \left|\val_{U}(\phi)-\val(\phi)\right|\le \varepsilon + \frac{k^2}{N} + \mathbb{I}[|U|\neq N] \quad\forall \phi\in \Psi 
\end{align}
where $\delta=2 |\Psi|  \exp(-\frac{\varepsilon^2 N}{2k^2}) = 2 q^{k^2s/\varepsilon}  \exp(-\frac{\varepsilon^2 N}{2k^2})$. 
This implies that with probability at least $1-\delta$
\[ \max_{\phi \in \Psi} \val_{U}(\phi) \le \OPT +  \varepsilon + \frac{k^2}{N} + \mathbb{I}[|U|\neq N] .\]
Indeed, to derive the inequality, it is sufficient to apply \Cref{eq:concSub} to the optimal assignment in $\Psi$ for $U$.
Then, using \Cref{clm:withToWithout}, we transition from sampling with replacement to sampling without replacement, yielding
\[ \mathbb{E} [\max_{\phi \in \Psi} \val_U(\phi) ] \le \OPT + \varepsilon+\frac{k^2}{N} + \delta + \frac{N^2}{n}. \]
Finally, by applying the structure lemma (\Cref{lm:structure}) we get

\begin{align*}
    \mathbb{E} [\OPT(U)] &\le \mathbb{E} [\max_{\phi \in \Psi} \val_U(\phi) ] + \mathcal{E}(\varepsilon,N,s) \\
    & \le  \OPT +\mathcal{E}(\varepsilon,N,s)+ \varepsilon+\frac{k^2}{N} + \delta + \frac{N^2}{n}.
\end{align*}

Setting $N=c\left(\frac{k}{\varepsilon}\right)^6\log^2(qk/\varepsilon)$ for a sufficiently large $c$, $s=c\frac{k^2}{\varepsilon^2}\log(qk/\varepsilon)$ and assuming $n\ge c^2\frac{k^{11}}{\varepsilon^{13}}\log^4(qk/\varepsilon)$, we get that $\mathcal{E}(\varepsilon,N,s)+ \varepsilon+\frac{k^2}{N} + \delta + \frac{N^2}{n}=O(k\varepsilon)$. Moreover, redefining $\varepsilon$ by $\varepsilon/k$ we get the desired result. 

\end{proof}

We note that the subsampling theorem only holds for $n\ge \left(\frac{k\log q}{\varepsilon}\right)^c$ for a suitable constant $c$, which, nonetheless, is the regime of interest, since otherwise we could just take $U=[n]$.

\section{Proof of \Cref{lem:concentration} (Concentration Lemma)} \label{sec:concentration}

Our proof relies on an algebraic representation of the CSP, i.e., its representation as a multilinear polynomial over a simplex. We embed the discrete assignments into a standard simplex, sample from this representation, and show that the sampled value concentrates around the true value. The analysis amounts to controlling two quantities: 1) the concentration error, which quantifies the deviation of the value of the sampled assignment from its expectation; and 2) the bias, i.e., the gap between the expected value of the polynomial and the value of the assignment on the whole CSP.

\begin{definition}[Algebraic version of a CSP] \label{def:alge}
    An assignment $\phi$ is represented by a vector $x(\phi)=(x_{i,a})_{i \in [n],a\in [q]}\in \mathbb{R}^{nq}$, where $x_{i,a}=\frac{1}{n}\mathbb{I}\{\phi_i=a\}$. Let $X=\{x(\phi)\}_{\phi \in [q]^n}$. For each constraint $w=(i_1,\ldots,i_k)\in\cW$ and constraint $P_w$, we define
    \[
    \tilde P_w(x)={n^{k}}\sum_{a_1,\dots,a_{k}\in[q]}P_w(a_1,\dots,a_{k})\prod_{j\in [k]} x_{i_j,a_j}.
    \]
    The objective function of the CSP is then extended to the polynomial $f:\Delta_{qn} \rightarrow [0,1]$ defined as:
    \[
    f(x)=\frac1{n^k}\sum_{w\in\cW} \tilde P_w(x).
    \]
\end{definition} 

 Note that $f(x)$ is a multilinear polynomial of degree $k$ satisfying $\sup_{x\in X} |f(x)|\le 1$. %
 Moreover, by definition we get $\OPT=\max_{x\in X}f(x)$.

Now, we can provide an equivalent representation of sampling a set of variables $U$.
Given an assignment $\phi$, sampling $N$ variables with replacement is equivalent to sampling $N$ times from the categorical distribution defined by $x(\phi)$. %
Let $(i_t, a_t)_{t=1}^N$ be $N$ i.i.d.~pairs sampled such that $\mathbb{P}[(i_t, a_t) = (i,a)] = x_{i,a}$. We define the empirical counts and mean as:
\[N_{i,a}:=\sum_{t=1}^N \mathbb{I}\{(i_t,a_t)=(i,a)\},\qquad
Z_{i,a}:=\frac{N_{i,a}}{N}.\]
We denote with $Z$ the empirical mean of these samples and with $\mathcal{D}(x, N)$ its distribution. 

We can now formally state the two quantities introduced informally at the beginning of the section.  In the following two sections, we show that the value of $f(Z)$ concentrates around its mean $\mathbb{E}[f(Z)]$ and that $\mathbb{E}[f(Z)]-f(x)$ is small, i.e., the bias is small. Then, the equivalence between the algebraic representation and the standard one concludes the proof. 

\subsection{Concentration}

Here, we show that $f(Z)$ concentrates around its expectation $\mathbb{E}[f(Z)]$. Formally:
\begin{lemma}\label{lem:concentration1}
     Fix $x\in X$, $N \in [n]$, and $\varepsilon>0$. Let $Z\sim \cD(x,N)$. With probability at least $1-2\exp\left(-\frac{\varepsilon^2N}{2k^2}\right)$, 
    \[
    |f(Z)-\mathbb{E}[f(Z)]|\le \varepsilon.
    \]
\end{lemma}

\begin{proof}
We start by proving that $f$ does not change too much when one sample gets changed. 
Fix a coordinate $(i,a)$ and a $\delta \in \mathbb R$. Take a constraint $w=(i_1,\ldots,i_k)\in\cW$ where variable $i$ appears at position $u\in[k]$. Let $e_{i,a}$ be the indicator function on index $(i,a)$. The multilinearity of $\tilde P_w(Z)$ implies that:
\begin{align*}
\left|\tilde P_w(Z+\delta e_{i,a})-\tilde P_w(Z)\right|
&= n^{k} \left|\sum_{a_{-u}\in[q]^{k-1}} P_w(a_1,\dots,a_{u-1},a,a_{u+1},\dots,a_k)\,
\delta \prod_{j\in [k], j\neq u} Z_{i_j,a_j}\right|\\
&\le n^{k} |\delta| \sum_{a_{-u}\in[q]^{k-1}} \prod_{j\in [k],j\neq u} Z_{i_j,a_j}\\
&= n^{k} |\delta| \prod_{j\in [k],j\neq u}\sum_{b\in[q]} Z_{i_j,b}.
\end{align*}
where the last equality is just repeated application of distributivity.
Summing over all constraints where $i$ occupies the $u$-th position, we get
\begin{align*}
\sum_{\substack{w=(i_1,\ldots,i_k)\in \cW:\\ i_u=i}}\,\,\prod_{j\neq u}\, \left(\sum_{b\in[q]} Z_{i_j,b}\right)
&\le \sum_{\substack{w=(i_1,\ldots,i_k)\in [n]^k:\\ i_u=i}}\,\,\prod_{j\neq u} \left(\sum_{b\in[q]} Z_{i_j,b}\right)\\
&=\sum_{(i_1,\dots,i_{u-1},i_{u+1},\dots,i_k)\in[n]^{k-1}}\,\,\prod_{j\neq u} \left(\sum_{b\in[q]} Z_{i_j,b}\right)\\
&= \left(\sum_{r\in[n]} \sum_{b\in[q]} Z_{r,b} \right)^{k-1}
=1.
\end{align*}
Therefore,
\[
\sum_{w=(i_1,\ldots,i_k)\in \cW:\, i_u=i}\bigl|\tilde P_w(Z+\delta e_{i,a})-\tilde P_w(Z)\bigr|
\le n^k|\delta|.
\]
Averaging over all $w$ and summing over all $k$ possible positions $u\in[k]$ for variable $i$ %
yields:
\begin{align}\label{eq:delta}
|f(Z+\delta e_{i,a})-f(Z)|
=\frac1{n^k}\left|\sum_{w\in \cW}\bigl(\tilde P_w(Z+\delta e_{i,a})-\tilde P_w(Z)\bigr)\right|
\le \frac1{n^k}\sum_{u=1}^k n^k|\delta|
= k|\delta|.
\end{align}

Changing one of the $N$ samples that define $Z$ changes exactly two coordinates by $\frac{1}{N}$. Then, by \Cref{eq:delta} and a triangle inequality, if we take any $Z$ and any $Z'$ obtained by changing one of the $N$ samples we get 
\[ 
|f(Z)-f(Z')|\le \frac{2k}{N}
\]
Applying the McDiarmid's inequality \citep{mcdiarmid1989method}, it is then guaranteed that
\[ 
\Big|f(Z)-\mathbb{E}[f(Z)]\Big|\le \varepsilon
\]
with probability at least $1-2\exp(-\frac{\varepsilon^2 N}{2k^2})$, concluding the proof.
\end{proof}

\subsection{Bounding the Bias}

Since $f$ is non-linear, the concentration of $f(Z)$ around $\mathbb{E}[f(Z)]$ does not by itself imply concentration around $f(\mathbb{E}[Z])=f(x)$. However, we bound this bias by exploiting the multilinear structure of $f$. In particular, we control the bias $ |\mathbb E[f(Z)]-f(x)|$ using known results on the factorial moments of the multinomial distribution (such as the one of \citet{mosimann1962compound, ouimet2021general}). 

\begin{claim}[\cite*{mosimann1962compound, ouimet2021general}]\label{lem:mosimann}
    Let $(N_j)_{j\in[d]} \sim \mathrm{Multinomial}(N, y)$, where $y\in \Delta_d$.  
    For any nonnegative integers $g_1,\dots,g_d$, let $g=\sum_{j=1}^d g_j$. 
    Then
    \[
    \mathbb E\left[\prod_{j=1}^d N_j^{\underline{g_j}}\right] = N^{\underline{g}} \prod_{i=1}^d y_i^{g_i}.
    \]
\end{claim}

By leveraging this result we can bound the bias as follows.

\begin{lemma}\label{lem:bias}
Fix $x \in X$, $N \in [n]$, and let $Z\sim \cD(x,N)$. Then
    \[
        |\mathbb E[f(Z)]-f(x)|\le \frac{k^2}{N}.
    \]
\end{lemma}

\begin{proof}
For any $w=(i_1,\ldots,i_k)\in \cW$, it holds
\begin{align}
\mathbb{E}\left[\tilde P_w(Z)\right]&= \mathbb{E}\left[n^{k}  \sum_{a_1,\dots,a_{k}\in[q]} P_w(a_1,\dots,a_{k})\prod_{j\in [k]} Z_{i_j,a_j}  \right]\notag\\
& =\mathbb{E}\left[ \left( \frac{n}{N}\right)^k \sum_{a_1,\dots,a_{k}\in[q]} P_w(a_1,\dots,a_{k})\prod_{j\in [k]} N_{i_j,a_j}\right]\notag\\
& =\left( \frac{n}{N}\right)^k \sum_{a_1,\dots,a_{k}\in[q]} P_w(a_1,\dots,a_{k}) \mathbb{E}\left[ \prod_{j\in [k]} N_{i_j,a_j}\right]\notag\\
& = \left( \frac{n}{N}\right)^k \sum_{a_1,\dots,a_{k}\in[q]} P_w(a_1,\dots,a_{k}) N^{\underline{{k}}}\prod_{j\in [k]} x_{i_j,a_j},\label{eq:bias}
\end{align}
where in the last equality we used \Cref{lem:mosimann}.

Hence, we get that 
\[
\mathbb{E}[f(Z)]=\frac{1}{N^k} \sum_{w=(i_1,\ldots, i_k)\in \cW} N^{\underline{k}}\sum_{a_1,\dots,a_{k}\in[q]} P_w(a_1,\dots,a_{k})\prod_{j\in [k]} x_{i_j,a_j}=\frac{N^{\underline{k}}}{N^k}f(x),
\]
where in the first equality we used \Cref{eq:bias}.

Thus, to conclude the proof, it is sufficient to observe that $\frac{N^{\underline{k}}}{N^k}\simeq 1$. Formally,
\begin{align*}
|\mathbb{E}[f(Z)]-f(x)| &= \left|\frac{N^{\underline{k}}}{N^k}-1\right| f(x)\le \frac{k^2}{N},
\end{align*}
where in the last inequality we used \Cref{lem:trivialtmp}.
\end{proof}

\subsection{Conclusion of the Proof of \Cref{lem:concentration}}

Having established the concentration of the polynomial $f(Z)$ and the bound on its bias, we now combine these results to prove the concentration lemma.

Let $\phi \in [q]^n$ be an assignment and let $x=x(\phi) \in X$ be its algebraic representation as defined in \Cref{def:alge}. 
Recall that $Z \sim \mathcal{D}(x, N)$ represents the empirical mean of $N$ variables sampled from $[n]$ with replacement according to the distribution $x$.

By the triangle inequality, the total error can be decomposed into concentration and bias terms:

\begin{equation}\label{eq:triangle}
|f(Z) - f(x)| \le |f(Z) - \mathbb{E}[f(Z)]| + |\mathbb{E}[f(Z)] - f(x)|.
\end{equation}
From \Cref{lem:concentration1}, with probability at least $1-2 \exp(-\frac{\varepsilon^2 N}{2k^2})$, it holds
\begin{equation*}\label{eq:f-conc}
|f(Z)-\mathbb E [f(Z)]| \le \varepsilon.
\end{equation*}
Moreover, by \Cref{lem:bias}, we get that
\begin{equation*}\label{eq:useBias}
|\mathbb{E}[f(Z)] - f(x)| \le  \frac{k^2}{N}.
\end{equation*}

Finally, by observing that $\val(\phi)=f(x)$ and $\val_{U}(\phi)= f(Z)$ if $|U|=N$ (i.e., there were no duplicates when sampling with replacement), we get that 
\[
\big| \val_U(\phi)-\val(\phi)\big| \le \varepsilon + \frac{k^2}{N} + \mathbb{I}[|U|\neq N]
\]
with probability at least $1-2 \exp(-\frac{\varepsilon^2 N}{2k^2})$. 
Notice that we used  $\mathbb{I}[|U|\neq N]$ to succinctly state that the bound holds only when $|U|=N$, by making it vacuous for $|U|\neq N$.

\section{Proof of \Cref{lm:structure} (Structure Lemma)}\label{sec:structure}
The proof structure follows the framework established by \citet{goldreich1998property} for {Max-CUT} and extended to general CSPs by \citet{barak2011subsampling}.
The goal is to demonstrate the existence of a set of assignments $\Psi$ that is simultaneously small enough to allow for a union bound over the concentration probability of each assignment, yet rich enough to approximate the optimum of a random subsampled instance.

We begin by partitioning the set of all variables $V \coloneqq [n]$ into $T$ blocks $\{V_t\}_{t\in [T]}$ and keep only ``balanced'' constraints, namely those containing at most one variable from each block.
By discarding constraints that fall within a single block, each block contribution is separable over its variables, but this does not change the objective much, as the removed constraints account for only a small fraction of the total mass.

Now, we sample a set $U_t$ of $N/T$ variables from each set $V_t$, and let $U=\bigcup_{t \in [T]} U_t$.
Then, for each block, we choose a smaller set $S_t\subseteq U \setminus U_t$ that lies outside each block $V_t$. The role of each $S_t$ (together with a guess assignment over its variables called ``seed'') is to summarize the interaction between block $U_t$ and the rest of the variables; instead of evaluating best-responses with respect to the whole set $U\setminus U_t$, we only use the much smaller set $S_t$.

The assignment is then reconstructed block by block through an iterative best-response procedure, allowing the estimate to adapt to the best-response assignments of previous blocks. In this way, the final assignment depends only on the guessed seed, rather than on the entire set $U$.
This leads to a ``deterministic extender'': for every possible assignment to the seed sets $\ccS=(S_t)_{t \in [T]}$ the extender outputs a full assignment on the original variables. The set of assignments generated by all the seeds provides our $\Psi$.

The above is the high-level idea of the proof; however, we have omitted many technical details. The technical proof is organized as described in the following paragraph. Before proceeding to the proof structure, it is convenient to define here a few probability distributions:
\begin{itemize}
    \item $\mu$ is the uniform distribution over subsets of $[n]$ of size $N$;
    \item $\gamma$ is the joint distribution of $(U,\ccS)$ sampled in this way: $U=\bigcup_{t\in[T]}U_t$ where $U_t$ (with $|U_t|=N/T$) is sampled without replacement uniformly from $V_t$, and then $S_t$ (with $|S_t|=s$) is sampled  without replacement uniformly from $U\setminus U_t$;
    \item $\nu$ is the marginal distribution of $U$ induced by the joint distribution $\gamma$. Equivalently, $U=\bigcup_{t\in[T]}U_t$ where $U_t$ (with $|U_t|=N/T$) is sampled without replacement uniformly from $V_t$;
    \item $\gamma_\ccS$ is the conditional distribution of $U$ under $\gamma$, given the realization $\ccS$.
\end{itemize}
With these definitions, we are now ready to outline the structure of the proof:

\begin{enumerate}
    \item We define the value of a ``balanced'' CSP, denoted with $\val^\bal_U$, which considers only those constraints where each of the $k$ variables belongs to a distinct block $V_t$ (\Cref{sec:balanced}).
    \item For every possible assignment on the seed sets $\ccS$, we construct a corresponding candidate set of assignments $\Psi(\ccS)$ (\Cref{sec:buildingPsi}).
    \item For a set of variables $U$ that contains $\ccS$, we upper bound the gap between $\OPT^\bal(U)$ (the optimal value of the balanced CSP on $U$) and $\max_{\phi \in \Psi(\ccS)} \val^\bal_U(\phi)$ (the best assignment in our candidate set). This bound depends on the sets $U$ and the tuple $\ccS$ (\Cref{sec:USapprox}).
    \item We show that the gap between $\OPT^\bal(U)$ and $\max_{\phi \in \Psi(\ccS)} \val^\bal_U(\phi)$ is small in expectation over the joint distribution $\gamma$ (\Cref{sec:randomized}).
    \item Using an averaging argument, we establish the existence of a deterministic seed tuple $\ccS^*$ that preserves this small gap, though $U$ is sampled from the conditional distribution $\gamma_{\ccS^*}$ (\Cref{sec:derandom}).
    \item We prove that for any family of assignments $\Phi\subseteq [q]^n$, the difference between the expected maximum values under the conditioned distribution $\gamma_{\ccS^*}$ and the uniform distribution $\mu$ is negligible. This allows us to transition back to $U \sim \mu$, showing that $\mathbb{E}_{\mu}[\OPT^\bal(U)] \approx \mathbb{E}_{\mu}[\max_{\phi \in \Psi(\ccS^*)}\val^\bal_U(\phi)]$ (\Cref{sec:uniformize}).
    \item We conclude by showing that the expected value of the balanced CSP is a sufficiently close proxy for the original CSP value under the uniform sampling measure $\mu$ (\Cref{sec:debalancing}).
\end{enumerate}

\subsection{Building a Balanced CSP}\label{sec:balanced}

To simplify the constraints between variables, we partition  the variable set $[n]$ into $T$ disjoint blocks $V_1, \ldots, V_T$, each of size approximately $n/T$. For any variable $i \in [n]$, let $b(i) \in [T]$ denote the index of the block containing $i$.

We focus our analysis on balanced constraints: those whose scope contains at most one variable from any block. Formally, let $\cM \subseteq \cW$ be the set of balanced $k$-tuples:
\[
\cM:=\{(i_1,\ldots,i_k)\in\cW:b(i_j)\neq b(i_{j'}),\,\forall j\neq j'\in[k]\}.
\]

Given a subset of variables $U \subseteq [n]$, we let $\cM(U) = \cM \cap U^k$ denote the balanced constraints restricted to $U$. For a specific variable $i$, we define $\cM_i(U) = \{w \in \cM(U) : i \in w\}$ as the set of balanced constraints in $U$ that include $i$ in their scope.

Then, we define the balanced CSP value restricted to $U$ as:
\[\val^\bal_U(\phi) = \frac{1}{|U|^k} \sum_{w \in \cM(U)} P_w(\phi|_w),\]
and its corresponding optimum as
\[ 
\OPT^{\bal}(U)=\max_{\phi \in [q]^U} \val^\bal_U(\phi).
\]
Now, we show how to decompose the value of an assignment $\phi$ over a subset of variables $U$ into two components. To do that, we use the fact that there are no constraints among variables of the same block.
Let 
\[
H_{i,U}(a, \psi)=\sum_{w\in \cM_i(\{i\}\cup U)}P_{w}((i\rightarrow a  )\otimes \psi|_{w\setminus \{i\}}).
\]
Intuitively, $H_{i,U}(a,\psi)$ denotes the total contribution of balanced constraints containing $i$ and otherwise supported on $U$, when $i$ is assigned to $a$, and the remaining variables are assigned according to $\psi$.

For a set $U$, let $U_t = U \cap V_t$ denote the variables in block $t$, and $U_{-t} = U \setminus U_t$ denote the variables in all other blocks. 
The following lemma (whose proof is deferred to \Cref{app:omitted}) partitions the value of the balance CSP into a term independent of $U_t$ and a sum of contributions from variables within $U_t$.

\begin{restatable}[Decomposition]{lemma}{lemmaDecomposition}\label{lem:decomposition}
    Consider a subset of variables $U\subseteq [n]$, an assignment $\phi\in [q]^{U}$, and fix $t \in [T]$. Then,
\begin{align*}
\val^{\bal}_U(\phi) & = \frac{1}{|U|^k}\sum_{w\in\cM(U_{-t})}P_{w}(\phi|_{w})+ \frac{1}{|U|^k} \sum_{i\in U_t}H_{i,U_{-t}}(\phi_i, \phi).
\end{align*}
\end{restatable}

At a high level, the proof of \Cref{lem:decomposition} works as follows. Since every constraint in $\cM(U)$ is balanced, it can contain at most one variable from $V_t$. Thus, any constraint $w \in \cM(U)$ either avoids $U_t$ entirely (belonging to $\cM(U_{-t})$) or intersects $U_t$ at exactly one variable $i$.

A critical consequence of this decomposition is that the optimal assignment for a variable $i \in U_t$ can be determined locally, i.e., independently from the assignment to variables in $U_{t} \setminus \{i\}$. Given a fixed assignment $\psi$ to all variables in $U \setminus \{i\}$, the value of $a$ that maximizes the global balanced CSP is identical to the value that maximizes $H$:
\begin{align}\label{eq:optimality}
    \arg\max_{a\in[q]}\, \val_U^\bal( (i\rightarrow a)\otimes \psi)=  \arg\max_{a\in[q]}\, H_{i, U_{-t}} (a,\psi).
\end{align}
This property will be of great importance. It will allow us to sequentially extend a partial assignment block-by-block as a function of $H$ only.

\subsection{Building $\Psi(\ccS)$}\label{sec:buildingPsi}

As an intermediate step, we construct a candidate set of assignments, $\Psi(\ccS)$, which depends on a tuple of variable subsets $\ccS = (S_t)_{t \in [T]}$. Recall that we impose $S_t\subseteq V\setminus V_t$. It is important to emphasize that, at this stage, the tuple $\ccS$ is arbitrary and independent of the set of variables $U$. We define the set $\Psi(\ccS)$ by including an assignment for each ``seed'' assignment $\bar \sigma = (\sigma_t)_{t \in [T]}\in \bigtimes_{t \in [T]} [q]^{S_t}$.
 In particular, for each seed $\bar \sigma$, we build an assignment using a ``deterministic extender'' $D_\ccS(\bar \sigma)$, which extends the seed assignment $\bar\sigma$ to a full assignment over all $V$. For any variable $i \in [n]$ with block index $t = b(i)$, the assignment $\phi_i$ is chosen as follows:
\begin{align}\label{eq:decoder}
\phi_i= D_\ccS(\bar \sigma)_i\in \arg\max_{a\in[q] }\val^\bal_{\{i\}\cup S_t}((i\rightarrow a) \otimes \sigma_t),
\end{align}
where ties are broken in an arbitrary but deterministic way.

The intuition behind this construction is rooted in the optimality condition established in \Cref{eq:optimality}. Since the optimal assignment for a variable in $V_t$ depends only on the values assigned to variables in $V \setminus V_t$, we can think about $S_t$ as a representative sample of $V \setminus V_t$.

\subsection{ A $(U,\ccS)$-Dependent Approximation}\label{sec:USapprox}

In this section, we still work with an arbitrary set $U$ and an arbitrary tuple $\ccS$ such that $S_t\subseteq U_{-t}$. Our objective is to bound the suboptimality of $\Psi(\ccS)$ as a function of the chosen sets. Then, in the next section, we will pick a distribution over $U$ and $\ccS$, and show that the bound is small in expectation over such distribution.

Let $\phi^*$ be the optimal assignment for the balanced CSP on $U$ (i.e., $\val_U^\bal(\phi^*) = \OPT^{\bal}(U)$). We define a sequence of assignments $\phi^{(0)}, \phi^{(1)}, \dots, \phi^{(T)}$ starting from $\phi^{(0)} = \phi^*$. Each $\phi^{(t)}$ is obtained by updating the variables in block $U_t$ to ``best respond'' to the previous assignment restricted to $S_t$:

\begin{align}\label{eq:seq_assign}
\phi^{(t)}_i = \begin{cases}
    \phi^{(t-1)}_i &\text{if}\quad i\in U_{-t}\\
    \in\arg\max_{a\in [q]}\val^\bal_{\{i\}\cup S_t}((i\rightarrow a) \otimes \phi^{(t-1)}|_{S_t})&\text{if}\quad i\in U_t,
\end{cases}
\end{align}
where ties are broken consistently with \Cref{eq:decoder}.

Crucially, $\phi^{(t)}$ depends only on the partial information $\{S_{t'}\}_{t' \in [t]}$ and the initial optimum $\phi^*$, and \emph{not} on the other sets $\{S_{t'}\}_{t' \in t+1,\ldots,T}$. This property will be useful for the concentration argument in the next section.

Moreover, $\phi^{(t)}$ is naturally related to the seed assignment $\sigma_t$, by setting $\sigma_t=\phi^{(t-1)}|_{S_t}$. While the extended assignment $D_\ccS(\bar\sigma)$ is difficult to analyze as it generates a complete assignment in ``one-shot'', working with the iterative construction of $\phi^{(t)}$ greatly simplifies our analyses. Indeed, this lets us bound the error between close assignments $\phi^{(t)}$ and $\phi^{(t-1)}$.

We formalize the connection between the deterministic extender, the seed, and the sequence of assignments $\phi^{(t)}$ in the following claim:

\begin{claim}\label{claim:recoversigma}
    The seed $\tilde \sigma = (\phi^{(t-1)}|_{S_t})_{t\in[T]}$ guarantees that
    \[
    D_\ccS(\tilde\sigma)|_U=\phi^{(T)}.
    \]
\end{claim}
\begin{proof}
    Let $i \in U$ and $t = b(i)$. By definition, the value of $\phi^{(T)}_i$ is set at step $t$ and remains unchanged thereafter. Specifically, $\phi^{(T)}_i = \phi^{(t)}_i \in \arg\max_{a\in [q]}\val^\bal_{\{i\}\cup S_t}((i\rightarrow a) \otimes \phi^{(t-1)}|_{S_t})$. This is exactly $\arg\max_{a\in[q] }\val^\bal_{\{i\}\cup S_t}((i\rightarrow a) \otimes \tilde \sigma_t)$ for $\tilde \sigma_t= \phi^{(t-1)}|_{S_t}$. 
\end{proof}

Now, we characterize the suboptimality of the assignment $\phi^{(T)}$. In particular, we show that the total error is the sum of $T$ error components, each related to how well the assignment over the small seed $S_t$ approximates the assignment over the entire $U_{-t}$.

\begin{lemma} \label{lm:dependentNet}
    Given $U\subseteq [n]$ and $\ccS = (S_t)_{t \in [T]}$ with $S_t \subseteq U_{-t}$, there exists a seed $\tilde \sigma=(\tilde \sigma_t)_{t \in [T]}$ and a sequence of assignments $\phi^{(t)}$ such that $\phi^{(t)}$ depends only on $\{S_{t'}\}_{t'\in [t]}$ and $U$ such that
    \[
    \OPT^\bal(U) - \val_{U}^\bal(D_\ccS(\tilde \sigma)) \le \sum_{t\in[T]} 2\frac{|U_{-t}|^{k-1}}{|U|^k}\sum_{i\in U_t}\max_{a\in[q]}\left|\frac{H_{i, U_{-t}}(a, \phi^{(t-1)}\mid_{U_{-t}})}{|U_{-t}|^{k-1}}-\frac{ H_{i, S_t}(a,\phi^{(t-1)}\mid_{S_t})}{|S_{t}|^{k-1}}\right|.
    \]
\end{lemma}

\begin{proof}

Consider the sequence of assignments $\phi^{(0)},\ldots, \phi^{(T)}$ defined in \Cref{eq:seq_assign}.

First, we observe that $\phi^{(t)}$ depends only on $\{S_{t'}\}_{t'\in [t]}$ and $U$ (which also defines the optimal assignment $\phi^*$).
Moreover, by \Cref{claim:recoversigma}, there exists a $\tilde \sigma$ such that:
\begin{align*}
\val^\bal_U(D_\ccS(\tilde\sigma))=\val^\bal_U(\phi^{(T)}).
\end{align*}
Hence, we are left to upper bound:
\[\OPT^\bal(U) - \val_{U}^\bal(\phi^{(T)}).\]
Let the error at step $t$ be 
\[\mathrm{err}(t) = \val^\bal_U(\phi^{(t-1)}) - \val^\bal_U(\phi^{(t)}).\]
Then, the total suboptimality gap is given by the following telescoping sum:
\[\OPT^\bal(U) - \val_U^\bal(\phi^{(T)})= \val_U^{\bal}(\phi^{(0)})-\val_U^{\bal}(\phi^{(T)}) = \sum_{t=1}^T \mathrm{err}(t). \]

Now, we want to bound the errors $\mathrm{err}(t)$ as a function of how well $H_{i,S_t}$ approximates $H_{i,U_{-t}}$. 
Notice that the crucial difference between the functions $H_{i,S_t}$ and $H_{i,U_{-t}}$ is the set of variables $i$ interacts with. Moreover, note that $\phi^{(t)}_i=\phi^{(t-1)}_i$ for all $i \in U_{-t}$.
Thus, thanks to \Cref{lem:decomposition}, we get
\begin{align*}
    \mathrm{err}(t)=\val^{\bal}_U(\phi^{(t-1)})-\val^{\bal}_U(\phi^{(t)})&=\frac{1}{|U|^k}\sum_{i\in U_t}\left( H_{i, U_{-t}}(\phi^{(t-1)}_i, \phi^{(t)})-H_{i, U_{-t}}(\phi^{(t)}_i , \phi^{(t)})\right),
\end{align*}
Now, fix a variable $i \in U_t$. For each $a\in [q]$, let 
\[A(a):= \frac{H_{i, U_{-t}}(a,\phi^{(t)})}{|U_{-t}|^{k-1}} \quad \textnormal{and}\quad  \hat A(a):=\frac{ H_{i, S_t}(a,\phi^{(t)})}{|S_{t}|^{k-1}}.\] 
Then:
\begin{align*}
    A(\phi^{(t-1)}_i)-A(\phi^{(t)}_i)&=\left(A(\phi^{(t-1)}_i)-\hat A(\phi^{(t-1)}_i)\right)+\left(\hat A(\phi^{(t-1)}_i)-\hat A(\phi_i^{(t)})\right)+\left(\hat A(\phi_i^{(t)})-A(\phi^{(t)}_i)\right)\\
    &\le |A(\phi^{(t-1)}_i)-\hat A(\phi^{(t-1)}_i)|+|\hat A(\phi_i^{(t)})-A(\phi^{(t)}_i)|\\
    &\le 2\max_{a\in[q]}|A(a)-\hat A(a)|,
\end{align*}
where the first inequality follows since $\phi^{(t)}_i$ is chosen specifically to maximize $\hat{A}$ (see \Cref{eq:optimality} and \Cref{eq:seq_assign}).

Then, we conclude that:

\[
\mathrm{err}(t)\le 2\frac{|U_{-t}|^{k-1}}{|U|^k}\sum_{i\in U_t}\max_{a\in[q]}\left|\frac{H_{i,U_{-t}}(a, \phi^{(t-1)})}{|U_{-t}|^{k-1}}-\frac{ H_{i,S_{t}}(a,\phi^{(t-1)})}{|S_{t}|^{k-1}}\right|.
\]

Summing over all the $t\in [T]$, we get the desired result.
\end{proof}

\subsection{From Deterministic to Randomized $\ccS$}\label{sec:randomized}

In the previous section, we established a deterministic bound on the suboptimality of the candidate set $\Psi(\ccS)$. We now show that if $U$ and $\ccS$ are sampled from the joint distribution $\gamma$ described at the beginning of \Cref{sec:structure}, this estimation error is small with high probability.

\begin{lemma}\label{lm:random}
    Let $(U,\ccS)\sim\gamma$, then with probability at least $1-2qN \exp\left(-\frac{\varepsilon^2 s}{2(k-1)^2}\right)-\frac{2Ts^2}{N}$, it holds:
    \[ 
    \OPT^\bal(U)- \max_{\phi \in \Psi(\ccS)}  \val^{\bal}_U(\phi) \le 2k\left(\varepsilon+\frac{(k-1)^2}{s}\right).
    \]
\end{lemma}

\begin{proof}
    Consider the sequence of assignments $\phi^{(0)}, \dots, \phi^{(T)}$ in \Cref{lm:dependentNet}.  Note that for a fixed $t$, the assignment $\phi^{(t-1)}$ is determined by the realizations of $U$ and the previous seeds $\{S_{t'}\}_{t' \in [t-1]}$.

    For every $t\in[T],i\in U_t,a\in[q]$, define:

    \[
    v^t_1(i,a) = \frac{H_{i,U_{-t}}(a, \phi^{(t-1)})}{k|U_{-t}|^{k-1}},\quad\text{and} \quad v^t_2(i,a) = \frac{H_{i,S_t}(a, \phi^{(t-1)})}{ks^{k-1}}.
    \]
    Notice that these quantities depend only on $U$ and $\{S_{t'}\}_{t'\in [t]}$.
    We observe that $v^t_1(i,a)$ is the value of a $(k-1)$-arity CSP on the variable set $U_{-t}$ evaluated at $\phi^{(t-1)}|_{U_{-t}}$, where the set of constraints is $\cM_i(\{i\}\cup U_{-t})$ (and so each constraint in the new CSP will be of range $[0,k]$, thus the normalizing factor $1/k$ in the definition of the values). Notice that the arity is $(k-1)$ since the value of the variable $i$ (which appears in all the constraints) is fixed to $a$.
    Moreover, $v^t_2(i,a)$ is the value of that same CSP restricted to constraints involving the subset of variables $S_t$.
 
    Here, we would like to apply our concentration lemma. However, $S_t$ is sampled \emph{without} replacement from $U_{-t}$.
    Hence, to apply the lemma, we first consider sampling a set $\tilde{S}_t$ of $s$ variables with replacement from $U_{-t}$, and define $\tilde{v}_2^t(i,a)$  as the value of the CSP on the subset of variables $\tilde S_t$, i.e.,
    \[
    \tilde v_2^t(i,a):=\frac{H_{i,\tilde S_t}(a,\phi^{(t-1)})}{k|\tilde S_t|^{k-1}}.
    \]

    Let $\delta=2 \exp\left(-\frac{\varepsilon^2 s}{2(k-1)^2}\right)$.
    By the concentration lemma (\Cref{lem:concentration}), for a fixed $(i,a)$, with probability at least $1-\delta$, we have:
    \[
    |v_1^t(i,a)-\tilde v_2^t(i,a)|\le \varepsilon+\frac{(k-1)^2}{s}+\mathbb{I}(|\tilde S_t|\neq s).
    \]
    By taking a union bound over all $i \in U_t$ and $a \in [q]$, we get that the event
    \[
    \tilde B_t:=\left\{\max_{i\in U_t,a\in [q]}\left|v_1^t(i,a)-\tilde v_2^t(i,a)\right|>\varepsilon+\frac{(k-1)^2}{s}+\mathbb{I}(|\tilde S_t|\neq s)\right\}
    \]
    holds with probability less than $\delta q|U_t|$ (again conditioned on the realizations of $U$ and $\{S_{t'}\}_{t'\in[t-1]}$). 
    
    Now, by \Cref{clm:withToWithout}, we extend the previous high-probability event to sampling without replacement and obtain that
    \[B_t:=\left\{\max_{i\in U_{t},a\in[q]}|v^t_1(i,a)-v^t_2(i,a)|>\varepsilon+\frac{(k-1)^2}{s}\right\}\] holds with probability at most $\delta q |U_{t}| + \frac{s^2}{|U_{-t}|}\le \delta q |U_{t}| + \frac{2s^2}{N}$ over the set $S_t$. Notice that $|U_{-t}|\ge N/2$ holds deterministically by our sampling choice of $U$.
    Then, by towering, we get 
    \[
    \mathbb{P}(\cup_{t\in[T]} B_t)\le \sum_{t\in[T]} \left(\delta q |U_{t}|+\frac{2s^2}{N}\right)= qN\delta + \frac{2Ts^2}{N},
    \]
    conditioned on the realization of $U$.
    
    Thus, conditioned on the realization of $U$, with probability at least $1-qN\delta-\frac{2Ts^2}{N}$, we obtain:
    \begin{align*}
    \OPT^\bal(U)- \max_{\phi \in \Psi(\ccS)}  \val^{\bal}_U(\phi)\le&\sum_{t\in[T]} 2k\frac{|U_{-t}|^{k-1}}{|U|^k}\sum_{i\in U_t}\max_{a\in[q]}\left|v_1^t(i,a)-v_2^t(i,a)\right|\tag{\Cref{lm:dependentNet}}\\
    &\le 2k\left(\varepsilon+\frac{(k-1)^2}{s}\right)\sum_{t\in[T]} \frac{|U_{-t}|^{k-1}|U_t|}{|U|^k}\\
    &\le 2k\left(\varepsilon+\frac{(k-1)^2}{s}\right),
    \end{align*}
    where in the last inequality we used that     
    \[
    \sum_{t=1}^T\frac{|U_{-t}|^{k-1}|U_t|}{|U|^k}=\sum_{t=1}^T\frac{(|U|-|U_t|)^{k-1}|U_t|}{|U|^k}\le \frac{|U|^{k-1}}{|U|^k}\sum_{t=1}^T |U_t|=1.
    \]
    Noticing that this holds uniformly over every realization of $U$ concludes the proof.
\end{proof}

\subsection{Derandomization}\label{sec:derandom}

In this section, we apply an averaging argument to derandomize the result of \Cref{lm:random}. Specifically, we prove the existence of a fixed $\ccS^*$ such that the approximation guarantees derived in the previous section hold in expectation over the conditional distribution $U \sim \gamma_{\ccS^*}$.

It will be useful to introduce an equivalent definition of $\gamma_{\ccS}$. For any subset $A \subseteq [n]$ such that $1 \le |A \cap V_t| \le N/T$ for all $t \in [T]$, let $\nu_A$ denote the law of a random set $U$ constructed as follows: for each block $V_t$, sample uniformly at random a subset $U_t$ of size $N/T$ such that $A \cap V_t \subseteq U_t$. We then define $U = \bigcup_{t \in [T]} U_t$. The following claim (whose proof is deferred to \Cref{app:omitted}) shows an equivalence between $\gamma_{\ccS}$ and $\nu_A$ for a specific set $A$.

\begin{restatable}{claim}{equivalencelaws}\label{claim:equivalencelaws}
    Consider any tuple $\ccS=(S_t)_{t\in[T]}$ and let $S=\bigcup_{t\in[T]}S_t$. Then the conditional law $\gamma_\ccS$ is identical to the law $\nu_S$.
\end{restatable}

Given a tuple $\ccS^* = (S^*_t)_{t \in [T]}$, let $S^* = \bigcup_{t \in [T]} S^*_t$ be its union. By invoking the equivalence established in \Cref{claim:equivalencelaws}, we can treat $\gamma_{\ccS^*}$ and $\nu_{S^*}$ interchangeably, omitting the explicit dependence of $S^*$ on $\ccS^*$ for convenience.

\begin{lemma} \label{lm:fixedS}
    There exists a tuple $\ccS^*=(S_t^*)_{t\in [T]}$, where $|S^*_t|=s$ for all $t\in[T]$, such that
    \begin{align*}%
        \mathbb{E}_{U\sim \nu_{S^*}}[\OPT^\bal(U)]
        \le
        \mathbb{E}_{U\sim \nu_{S^*}}\left[\max_{\phi\in \Psi(\ccS^*)} \val_U^{\bal}(\phi)\right]
        + \mathcal{E}_1(\varepsilon,N,s,T),
    \end{align*}
    where $\mathcal{E}_1(\varepsilon,N,s,T)=2k\left(\varepsilon+\frac{(k-1)^2}{s}\right)+ 2qN\exp\left(-\frac{\varepsilon^2 s}{2(k-1)^2}\right)+ \frac{2Ts^2}{N} $.
\end{lemma}
\begin{proof}
    Let $\alpha:=2k\left(\varepsilon+\frac{(k-1)^2}{s}\right)$ and $\delta:=2qN\exp\left(-\frac{\varepsilon^2 s}{2(k-1)^2}\right)+\frac{2Ts^2}{N}$.
    Define
    \[
    A(\ccS,U):=
    \mathbb{I}\left[
    \OPT^\bal(U)-\max_{\phi\in\Psi(\ccS)}\val_U^\bal(\phi)\le \alpha
    \right].
    \]
    By \Cref{lm:random}, we know that under the joint law of $(U,\ccS)$, $\mathbb{P}_{(U,\ccS)\sim\gamma}(A(\ccS,U)=1)\ge 1-\delta$. 
    
    By a standard averaging argument, there must exist at least one realization $\ccS^*$ such that:
            \[
            \mathbb{P}_{U\sim \gamma_{\ccS^*}}(A(\ccS^*,U)=1)\ge 1-\delta,
            \]
    where $\gamma_{\ccS}$ is the conditioning of the distribution in \Cref{lm:random} to the realization of $\ccS$.
    
    Since $0\le\OPT^\bal(U)-\max_{\phi\in\Psi(\ccS)}\val_U^\bal(\phi)\le 1$, by taking the expectation we get
    \[
    \mathbb{E}_{U\sim \gamma_{\ccS^*}}
    \left[
    \OPT^\bal(U)-\max_{\phi\in\Psi(\ccS^*)}\val_U^\bal(\phi)
    \right]
    \le \alpha+\delta.
    \]
    Finally, by \Cref{claim:equivalencelaws}, the law $\gamma_{\ccS^*}$ is equal to the one of $\nu_{S^*}$, with $S^*=\bigcup_{t\in[T]} S_t^*$, concluding the proof.
\end{proof}

\subsection{Turn $\nu_{S^*}$ into a Uniform Law}\label{sec:uniformize}

Now we would like to obtain a result similar to \Cref{lm:fixedS} for a uniform subset of variables $U$. Formally, we will show that:

\begin{lemma}\label{lm:finalUniform}
    There exists a set of assignments $\Psi\subseteq [q]^n$ of size $|\Psi|=q^{s T}$ such that
    \[
    \mathbb{E}_{U\sim \mu}[\OPT^\bal(U)]\le
    \mathbb{E}_{U\sim \mu}\left[\max_{\phi\in\Psi} \val_U^{\bal}(\phi)\right]
    + \mathcal{E}_1(\varepsilon,N,s,T)+  \mathcal{E}_2(\varepsilon,N,s,T),
    \]
    where $\mathcal{E}_2(\varepsilon,N,s,T)=\frac{2sTk}{N}+ k \sqrt{\frac{T}{N}}$ and $\mathcal{E}_1$ is defined as per \Cref{lm:fixedS}.
\end{lemma}

Before proving \Cref{lm:finalUniform} we will prove some intermediary lemmas.
Moreover, for brevity, we denote $\Psi(\ccS^*)$ simply as $\Psi$ and define the shorthand for the maximum value over an assignment family $\Phi\subseteq [q]^n$:
\[\OPT^\bal_\Phi(U) := \max_{\phi \in \Phi} \val_U^\bal(\phi).\]

We begin by showing that swapping a single variable in $U$ has a negligible effect on the optimal value $\OPT^\bal_{\Phi}(\cdot)$.
Let $A \triangle B= (A\cup B) \setminus (A \cap B)$ be the symmetric difference between $A$ and $B$. Then
\begin{lemma}\label{lm:lip}
    Let $\Phi \subseteq [q]^n$ and $U,U'\subseteq [n]$ be subsets of size $N$ such that $|U\triangle U'|= 2$. Then
    \[
    |\OPT^\bal_{\Phi}(U)-\OPT^\bal_{\Phi}(U')|\le \frac kN.
    \]
\end{lemma}
\begin{proof}
    Fix any assignment $\phi\in \Phi$. Since $|U\triangle U'|= 2$, there exists exactly an index $i\in U\setminus U'$ and an index $j\in U'\setminus U$. The number of constraints including $i$ is at most $k(N-1)^{k-1}$, and the number of constraints including $j$ is at most $k(N-1)^{k-1}$.
    Hence, %
    \begin{align*}
        |\val^\bal_U(\phi)-\val_{U'}^\bal(\phi)|\le \left|\frac{1}{N^k}\sum_{w\in \mathcal{M}(U)} P_w(\phi|_w) -  \frac{1}{N^k}\sum_{w\in \mathcal{M}(U')} P_w(\phi|_w)\right|  \le \frac{k(N-1)^{k-1}}{N^k}\le \frac{k}{N}.
    \end{align*}

    Since this holds for every $\phi$, it must also hold for the maximum over $\Phi$, i.e., a $\phi^*\in\arg\max_{\phi\in \Phi}\val^\bal_U(\phi)$. Hence, we get
    \[
    \max_{\phi\in \Phi}\val^\bal_U(\phi)=\val_{U}^\bal(\phi^*)\le \val_{U'}^\bal(\phi^*)+ \frac kN \le \max_{\phi\in \Phi}\val_{U'}^\bal(\phi)+ \frac kN.
    \]
    and also vice versa by taking a maximizer for $\val_{U'}^\bal$.
\end{proof}

Now, we show that for $U\sim \nu$, conditioning on a variable $i\in U$ does not change the value of the CSP much. For technical reasons, it is convenient to strengthen the statement by considering a $U\sim\nu_{A}$ where $A$ is a possible empty subset of $[n]$. Formally:

\begin{lemma}\label{lem:balancedUcontainsi}
    Let $\Phi \subseteq [q]^n$, $A\subseteq [n]$, and $i\in[n]$. Then:
    \[
    \left|\mathbb{E}_{\nu_A}[\OPT^\bal_\Phi(U)|i\in U]-\mathbb{E}_{\nu_A}[\OPT^\bal_\Phi(U)]\right|\le \frac kN.
    \]
\end{lemma}

\begin{proof}
    Let $\nu^+$ be the law of $U\sim \nu_A$ conditioned on $i\in U$, and let $\nu^-$ be the law of $U\sim \nu_A$ conditioned on $i\notin U$.
    Sample $U\sim \nu^+$. Let $t=b(i)$ (hence $i\in U_t$) and define
    \[
    U':=(U\setminus\{i\})\cup\{j\},
    \]
    where $j$ is sampled uniformly from $V_t\setminus U_t$.

    Since under $\nu_A$ only the block $V_t$ is affected by conditioning on $i\in U$, it is immediate that $U'\sim \nu^-$.
    Moreover, deterministically $|U\triangle U'|= 2$. Therefore, by \Cref{lm:lip},
    \[
        \left|\mathbb{E}_{U\sim \nu^+}[\OPT^\bal_\Phi(U)]-\mathbb{E}_{U'\sim \nu^-}[\OPT^\bal_\Phi(U')]\right|\le \frac{k}{N}.
    \]
    Finally,
    \begin{align*}
        &\left|\mathbb{E}_{U\sim \nu_A}[\OPT^\bal_\Phi(U)| i\in U]-\mathbb{E}_{U\sim \nu_A}[\OPT^\bal_\Phi(U)]\right|\\
        &=
        \left|
        \mathbb{E}_{U\sim \nu_A}[\OPT^\bal_\Phi(U)| i\in U]
        -\mathbb{P}(i\in U)\mathbb{E}_{U\sim \nu_A}[\OPT^\bal_\Phi(U)| i\in U]
        -\mathbb{P}(i\notin U)\mathbb{E}_{U\sim \nu_A}[\OPT^\bal_\Phi(U)| i\notin U]
        \right|\\
        &=
        \mathbb{P}(i\notin U)
        \left|
        \mathbb{E}_{U\sim \nu_A}[\OPT^\bal_\Phi(U)| i\in U]
        -
        \mathbb{E}_{U\sim \nu_A}[\OPT^\bal_\Phi(U)| i\notin U]
        \right|\le \frac{k}{N},
    \end{align*}
    concluding the proof.
\end{proof}

\begin{lemma} \label{lm:toBalanced}
    Let $\Phi \subseteq [q]^n$ and let $S\subseteq [n]$ with $|S|=r$.
    Then
    \[
    \left|\mathbb{E}_{U\sim \nu}[\OPT^\bal_\Phi(U)]-\mathbb{E}_{U\sim \nu_S}[\OPT^\bal_{\Phi}(U)]\right|\le \frac{rk}{N}.
    \]
\end{lemma}

\begin{proof}
    Define the intermediate measures $\nu_j:=\nu_{\{i_1,\ldots,i_j\}}$ for $j\in\{1,\ldots,r\}$, where
    $S=\{i_1,\ldots,i_r\}$. Moreover, let $\nu_0=\nu$.

    By the tower property,
    \[
    \mathbb{E}_{U\sim \nu_j}[\OPT^\bal_\Phi(U)]=    \mathbb{E}_{U\sim \nu_{j-1}}[\OPT^\bal_\Phi(U)\mid i_j\in U]
    \qquad \forall j\in[r].
    \]
    Therefore,
    \begin{align*}
        \left|\mathbb{E}_{U\sim \nu}[\OPT^\bal_\Phi(U)]-\mathbb{E}_{U\sim \nu_S}[\OPT^\bal_\Phi(U)]\right|
        &\le \sum_{j=1}^r
        \left|
        \mathbb{E}_{U\sim \nu_{j-1}}[\OPT^\bal_\Phi(U)]-
        \mathbb{E}_{U\sim \nu_{j-1}}[\OPT^\bal_\Phi(U)\mid i_j\in U]
        \right|\\
        &\le \sum_{j=1}^r \frac{k}{N}
        =\frac{rk}{N},
    \end{align*}
    where the second inequality follows from \Cref{lem:balancedUcontainsi} with $A=\{i_1,\ldots, i_{j-1}\}$.
\end{proof}

Now, we can finally extend \Cref{lm:fixedS} to $U\sim\nu$. Thus, we get that the optimal assignment in $\Psi$ is close in expectation to the optimal one.

\begin{lemma}\label{lm:finalBalanced}
    There exists a set of assignments $\Psi\subseteq [q]^n$ of size $|\Psi|=q^{s T}$ such that
    \[
    \mathbb{E}_{U\sim \nu}[\OPT^\bal(U)]\le
    \mathbb{E}_{U\sim \nu}\left[\max_{\phi\in\Psi} \val_U^{\bal}(\phi)\right]
    +\mathcal{E}_{1}(\varepsilon, N,s,T) 
    + \frac{2sTk}{N}.
    \]
\end{lemma}

\begin{proof}
    First, notice that \Cref{lm:fixedS} implies $|\Psi|=q^{sT}$.
    Now, since $|\cup_{t\in[T]}S_t^*|\le sT$, by \Cref{lm:toBalanced} we get
    \[
    \mathbb{E}_{U\sim \nu}[\OPT^\bal(U)]\le\mathbb{E}_{U\sim \nu_{S^*}}[\OPT^\bal(U)] +\frac{sTk}{N}.
    \]
    Applying \Cref{lm:fixedS}, we obtain
    \begin{align*}
        \mathbb{E}_{U\sim \nu}[\OPT^\bal(U)]
        &\le \mathbb{E}_{U\sim \nu_{S^*}}\left[\max_{\phi\in\Psi}\val_U^\bal(\phi)\right] + \mathcal{E}_{1}(\varepsilon, N,s,T)+ \frac{sTk}{N}\\
        &\le\mathbb{E}_{U\sim \nu}\left[\max_{\phi\in\Psi}\val_U^\bal(\phi)\right] + \mathcal{E}_{1}(\varepsilon, N,s,T) + \frac{2sTk}{N},
    \end{align*}
    where in the last inequality we applied \Cref{lm:toBalanced} again to the family $\Phi=\Psi$. 
\end{proof}

Now we have to argue that the expectation with respect to $\nu$ is close to the expectation with respect to the uniform distribution $\mu$. This again relies on  \Cref{lm:lip}.

\begin{lemma}\label{lm:toUni}
Let $\Phi\subseteq [q]^n$ be any family of assignments. Then
\[
    \left| \mathbb{E}_{U\sim\mu}[\OPT^\bal_\Phi(U)]-\mathbb{E}_{U\sim\nu}[\OPT^\bal_\Phi(U)]\right|\le \frac{k}{2}\sqrt{\frac{T}{N}}.
\]
\end{lemma}
\begin{proof}

    Let us define the block counter $M_t(U) = |U \cap V_t|$ and, similarly, the count vector $M(U) = (M_1(U), \dots, M_T(U))$. For each count vector $m = (m_1, \dots, m_T)$, we define the conditional expectation 
    \[g(m) = \mathbb{E}_{U \sim \mu}[\OPT_\Phi^\bal(U) \mid M(U) = m].\]
    Moreover, it will be convenient to define the perfectly balanced counter as $\bar{m} = (N/T, \dots, N/T)$.
    
    By the tower property, we obtain
    \begin{align}
    \mathbb{E}_{U\sim \mu}[\OPT_\Phi^\bal(U)] = \mathbb{E}_{U\sim \mu}[\mathbb{E}_{U\sim \mu}[\OPT_\Phi^\bal(U)|M(U)]]=\mathbb{E}_{U\sim \mu}[g(M(U))].\label{eq:tmp1234}
    \end{align}

    Furthermore, because $\nu$ can be interpreted as the uniform distribution conditioned on perfect balance ($M(U) = \bar{m}$), it is clear that
    \begin{align}
    \mathbb{E}_{U\sim \nu}[\OPT^\bal_\Phi(U)]=g(\bar m).\label{eq:tmp12345}
    \end{align}

    Now, let $m$ and $m'$ be two count vectors that differ by a single swap between blocks, meaning $m' = m + e_j - e_k$ for two distinct indices $j, k \in [T]$. We can couple the two distributions conditioned on $M(U)=m$ and $M(U)=m'$ by sampling $U$ from the distribution $\mu$ conditioned on  $M(U) = m$, picking $a$ uniformly at random from $U \cap V_k$ and $b$ uniformly at random from $V_j \setminus U$, and defining $U' = (U \setminus \{a\}) \cup \{b\}$.

    By construction, $U'$ is sampled from the distribution $\mu$ conditioned on $M(U) = m'$ and the symmetric difference of the coupling $|U \triangle U'| = 2$ deterministically. Hence, by \Cref{lm:lip}, we obtain the deterministic bound
    \[
    |\OPT^\bal_\Phi(U)-\OPT^\bal_\Phi(U')|\le \frac kN.
    \]
Taking the expectation over our coupled variables yields   
\[
    |g(m)-g(m')|\le \frac kN.
\]

Consider now an arbitrary count vectors $m$ and the perfectly balanced counter $\bar m$. By iterating this swap process over $\frac{1}{2} \|m - \bar{m}\|_1$ steps and applying the triangle inequality we get

\begin{align}\label{eq:gNorm}
    |g(m)-g(\bar m)|\le \frac{k}{2N}\|m-\bar m\|_1.
\end{align}

    Now, we can bound the difference:
    \begin{align*}
        \left|\mathbb{E}_{U\sim \mu}[\OPT_\Phi^\bal(U)]-\mathbb{E}_{U\sim \nu}[\OPT_\Phi^\bal(U)]\right|&=\left|\mathbb{E}_{U\sim \mu}[g(M(U))]-g(\bar m)\right|\tag{\Cref{eq:tmp1234,eq:tmp12345}}\\
        &\le \mathbb{E}_{U\sim \mu}\left[\left|g(M(U))-g(\bar m)\right|\right]\tag{Jensen}\\
        &\le \frac{k}{2N}\mathbb{E}_{U\sim \mu}[\|M(U)-\bar m\|_1]. \tag{\Cref{eq:gNorm}}
    \end{align*}

    Then, observe that each block counter $M_t(U)$ follows a hypergeometric distribution with mean $N/T$. Hence, its variance is upper bounded by $N/T$. Thus:
    \begin{align*}
    \mathbb{E}_{U\sim \mu}[\|M(U)-\bar m\|_1]&\le \sum_{t\in [T]}\mathbb{E}_{U\sim \mu}[|M_t(U)-N/T|]\\
    &\le T \sqrt{\frac NT}=\sqrt{NT},
    \end{align*}
    Substituting this back into the previous inequality gives
    \[
    \left|\mathbb{E}_{U\sim \mu}[\OPT_\Phi^\bal(U)]-\mathbb{E}_{U\sim \nu}[\OPT_\Phi^\bal(U)]\right|\le \frac{k}{2}\sqrt\frac TN,
    \]
    concluding the proof.
\end{proof}

Now it is straightforward to prove \Cref{lm:finalUniform}.

\begin{proof}[Proof of \Cref{lm:finalUniform}]
It is sufficient to take \Cref{lm:finalBalanced} and apply \Cref{lm:toUni} to both expectations. This increases the error by an additive error of ${k}\sqrt\frac TN$.
\end{proof}

\subsection{De-Balancing}\label{sec:debalancing}

Now, we want to conclude the proof by extending \Cref{lm:finalBalanced} from the balanced-CSP setting to the original CSP. To do that, we show that with high probability over $U\sim \mu$, the values of the two CSPs are close.

\begin{lemma}\label{lm:highProbBalance}
With probability at least $1-T \exp(-N/T^2)$  over $U\sim\mu$, it holds 
\[|\val_U(\phi)-\val^\bal_U(\phi)| \le \frac{4k^2}{T} \quad  \forall \phi \in [q]^U.\]
\end{lemma}

\begin{proof}

We start showing that, with high probability, $U$ is balanced, i.e., all $V_t \cap U_t$ are roughly of equal size. 
By a standard Hoeffding inequality \citep{dubhashi2009concentration}, the probability that $|U_t|>2N/T$ is smaller than $\exp(-N/T^2)$. By taking an union bound over all $t\in [T]$, we get that the $M=\max_{t\in [T]} |U_t|\le 2N/T$ with probability at least $1-T\exp(-N/T^2)$. 

In the following, we assume that $M\le 2N/T$ and hence our lemma will hold with probability $1- T \exp(-N/T^2)$.

Define
\[
\mathrm{Bad}(U)\coloneqq\{(i_1,\ldots, i_k)\in \cW(U):\exists i_j\neq i_{j'} \quad s.t.\quad b(i_j)=b(i_{j'})\}.
\]
Moreover, given a $t \in [T]$, let
\[
\mathrm{Bad}_t(U)\coloneqq\{(i_1,\ldots, i_k)\in \cW(U):\exists i_j\neq i_{j'} \quad s.t.\quad b(i_j)=b(i_{j'})=t\}.
\]
Then, we can easily upperbound 
\[|\mathrm{Bad}(U)|\le \sum_{t \in [T]} |\mathrm{Bad}_t(U)|\le   k^2 T M^2 N^{k-2}, \]
where the second inequality follows by observing that $\mathrm{Bad}_t(U)$ includes one of the at most $M$ variables in at least $2$ of the $k$ indexes. 

Then, given an assignment $\phi$, we get 
\[ 
|\val_U(\phi)-\val_U^\bal(\phi)|\le \frac{|\mathrm{Bad}(U)|}{N^k}\le \frac{k^2T M^2}{N^{2}}\le \frac{k^2T}{N^2} \frac{4 N^2}{T^2} = \frac{4k^2}{T},
\]

concluding the proof.
\end{proof}

\subsection{Conclusion of the Proof of \Cref{lm:structure}}
Now, we can finally prove \Cref{lm:structure}. We get:
\begin{align*}
    \mathbb{E}_\mu[\OPT(U)]  &\le \mathbb{E}_{\mu} [\OPT^{\bal}(U)] + \frac{4k^2}{T} + T \exp\left(-\frac{N}{T^2}\right)\\
    & \le \mathbb{E}_{\mu}[ \max_{\phi \in \Psi} \val^\bal_U(\phi)] + \mathcal{E}_1(\varepsilon,N,s,T)+  \mathcal{E}_2(\varepsilon,N,s,T)+  \frac{4k^2}{T} + T \exp\left(-\frac{N}{T^2}\right)\\
    & \le  \mathbb{E}_{\mu}[ \max_{\phi \in \Psi} \val_U(\phi)] + \mathcal{E}_1(\varepsilon,N,s,T)+  \mathcal{E}_2(\varepsilon,N,s,T) +  \frac{4k^2}{T} + T \exp\left(-\frac{N}{T^2}\right)\\
    & = \mathbb{E}_{\mu}[ \max_{\phi \in \Psi} \val_U(\phi)] + 2k\left(\varepsilon+\frac{(k-1)^2}{s}\right)
    + 2qN\exp\left(-\frac{\varepsilon^2 s}{2(k-1)^2}\right)\\
    &\hspace{3.3cm}+ \frac{2Ts^2}{N}
    + \frac{2sTk}{N} 
    + k \sqrt\frac{T}{N} +  \frac{4k^2}{T} + T  \exp\left(-\frac{N}{T^2}\right),
\end{align*}
where the first inequality follows from \Cref{lm:highProbBalance} (by also going from high probability to expectation), the second by \Cref{lm:finalUniform}, third one by observing that the balance objective has fewer constraints than the original one, and the last equality by the definitions of $\mathcal{E}_1$ and $\mathcal{E}_2$.
The proof is concluded by setting $T=k^2/\varepsilon$.

\section{Applications}

In this section, we will show two main implications of our result. In particular, we will show how our subsampling theorem provides the main piece to complete the proof of \citet{aaronson2014multiple} that $\AM(\poly)=\AM$, and that it provides a $\poly(\frac{k\log q}{\varepsilon})$ tester for satisfiability.

\subsection{Free Games}
We show how to use our subsampling theorem to subsample $k$-player free games (see \Cref{def:fg}). It is shown by \citet[Theorem~53]{aaronson2014multiple} that this is sufficient to establish the equivalence between $\mathsf{AM}(k)$ and $\mathsf{AM}$, for any polynomial $k$. The following proof follows almost verbatim from \citet{aaronson2014multiple}, except for subtleties in normalization factors, and we report it here only for the sake of completeness.

\begin{theorem}[Subsampling of $k$-Player Free Games]\label{thm:subsampling-k-free-games}
There exists a constant $c>0$ such that the following holds. For any $\varepsilon>0$ and any $k$-player free game $G$ with $|Y_1|\cdots|Y_k| \ge  c^2 \frac{k^{24}}{\varepsilon^{13}}\log^4(k^2|B_1|\cdots |B_k|/\varepsilon)$, let
\[
N := c\frac{k^{12}}{\varepsilon^6} \log^2(k^2|B_1|\cdots |B_k|/\varepsilon).
\]
For each $i\in[k]$, choose a tuple $S_i$ of size $|S_i|=N$ by sampling (with replacement) elements uniformly at random from $Y_i$, let $S:=S_1\times\cdots\times S_k$, and let $G_S$ be the subgame of $G$
where the questions of the $i$-th player are uniformly sampled from $S_i$. Then
\[
\mathbb{E}_{S}\bigl[\omega(G_S)\bigr]\le  \omega(G) + O(\varepsilon).
\]
\end{theorem}
\begin{proof}
Let $\mathbf{Y} := Y_1\times \dots \times Y_k$ and $\mathbf{B} := B_1\times \dots \times B_k$. We define an assignment as a function $\mathbf{b}:\mathbf{Y} \to \mathbf{B}$ and, for every $k$-tuple $\mathbf R=(\mathbf y_1,\dots,\mathbf y_k)\in \mathbf Y^{\underline{k}}$ of distinct elements, we define the constraint
$C_\mathbf{R}$ as
\[
C_\mathbf{R}(\mathbf{b}) 
:=
\mathbb{E}_{\sigma\in \mathrm{Sym}_k}\Bigl[
V\bigl(
(\mathbf y_{\sigma(1)})_1,\ldots,(\mathbf  y_{\sigma(k)})_k,\;
(\mathbf b(\mathbf y_{\sigma(1)}))_1,\ldots,(\mathbf b(\mathbf y_{\sigma(k)}))_k
\bigr)
\Bigr],
\]
where $\mathrm{Sym}_k$ is the symmetric group acting on $k$ elements (i.e., the set of all $k$-elements permutations). This defines a $k$-CSP with $|\mathbf{Y}|$ variables, $|\mathbf{Y}|^{\underline k}$ constraints, and alphabet size $|\mathbf{B}|$.%

Fix an arbitrary assignment $\mathbf{b} : \mathbf{Y}\to \mathbf{B}$. For each $i \in [k]$, define the distribution $D_i$ over functions $b_i : Y_i \to B_i$ obtained by sampling, independently for each input $y_i\in Y_i$, $y_j \sim Y_j$ uniformly at random for all $j \neq i$, and then setting the function to be $b_i(y_i) = \mathbf{b}(y_1, \dots, y_k)_i$. Then,
\begin{align}\label{eq:inverse}
    \E_{\mathbf{R}}\left[C_\mathbf{R}\left(\mathbf{b}\right)\right]
    &= \E_{(\mathbf{y}_1, \dots, \mathbf{y}_k) \sim \mathbf{Y}^{\underline{k}}}\left[
    V\left(
(\mathbf y_{1})_1,\ldots,(\mathbf  y_{k})_k,\;
(\mathbf b(\mathbf y_{1}))_1,\ldots,(\mathbf b(\mathbf y_{k}))_k
\right)\right] \tag{Symmetry}\\
    &\leq \E_{\mathbf{y}_i \sim \mathbf{Y}}\left[
    V\left(
(\mathbf y_{1})_1,\ldots,(\mathbf  y_{k})_k,\;
(\mathbf b(\mathbf y_{1}))_1,\ldots,(\mathbf b(\mathbf y_{k}))_k
\right)\right]
    + \frac{k^2}{|\mathbf Y|}\tag{Collision}\\
&= \E_{{y}_i \sim Y_i, b_i\sim D_i}
    \left[ V\left( y_1, \dots, y_k, b_1(y_1),\dots, b_k(y_k)
\right)\right]
    + \frac{k^2}{|\mathbf Y|}\tag{Definition of $D_i$}\\
    & \leq \omega(G)  + \frac{k^2}{|\mathbf Y|}. \tag{\Cref{def:fg}}
\end{align}
Let $\mathbf I = (\mathbf z_1,\ldots,\mathbf z_N) \sim \mathrm{Unif}\big(\mathbf Y^{\underline{N}}\big)$ be a subsample for the CSP as defined above. Applying 
\Cref{thm:subsampling} (with $n = |\mathbf{Y}|$) we obtain

\begin{align}
\frac{N^{\underline{k}}}{N^k} \E_{\mathbf{I}}\left[\max_{\mathbf{b}:\mathbf{Y} \to \mathbf{B}} \E_{\mathbf R\sim \mathbf I^{\underline{k}}}\left[C_{\mathbf R}(\mathbf b)\right]\right]&=
\E_{\mathbf{I}}\left[\max_{\mathbf{b}:\mathbf{Y} \to \mathbf{B}}\val_{\mathbf{I}}(\mathbf{b})\right]\notag\\
&\le \max_{\mathbf{b}:\mathbf{Y} \to \mathbf{B}} \val(\mathbf{b})+O(\varepsilon)\notag\\
&\le \frac{|\mathbf{Y}|^{\underline{k}}}{|\mathbf{Y}|^k} \max_{\mathbf{b}:\mathbf{Y} \to \mathbf{B}} \E_\mathbf{R}\left[C_\mathbf{R}(\mathbf{b})\right]+O(\varepsilon)\label{eq:main_thm}
\end{align}

and consequently
\begin{align}\label{eq:real_main}
\E_\mathbf{I}\left[\max_{\mathbf{b}:\mathbf{I} \to \mathbf{B}}\E_{\mathbf R\sim \mathbf I^{\underline{k}}}\left[C_\mathbf{R}\left(\mathbf{b}\right)\right] \right]-\frac{k^2}{N}&\le \left(1-\frac{k^2}{N}\right)\E_\mathbf{I}\left[\max_{\mathbf{b}:\mathbf{I} \to \mathbf{B}}\E_{\mathbf R\sim \mathbf I^{\underline{k}}}\left[C_\mathbf{R}\left(\mathbf{b}\right)\right] \right]\tag{$C_\mathbf{R}\le 1$}\\
&\le \frac{N^{\underline{k}}}{N^k}\E_\mathbf{I}\left[\max_{\mathbf{b}:\mathbf{I} \to \mathbf{B}}\E_{\mathbf R\sim \mathbf I^{\underline{k}}}\left[C_\mathbf{R}\left(\mathbf{b}\right)\right] \right]\tag{\Cref{lem:trivialtmp}}\\
&\leq \frac{|\mathbf{Y}|^{\underline{k}}}{|\mathbf{Y}|^{k}}\max_{\mathbf{b}:\mathbf{Y} \to \mathbf{B}}\E_{\mathbf{R}}\left[C_\mathbf{R}\left(\mathbf{b}\right)\right] + O(\varepsilon)\tag{\Cref{eq:main_thm}}\\
&\le \max_{\mathbf{b}:\mathbf{Y} \to \mathbf{B}}\E_{\mathbf{R}}\left[C_\mathbf{R}\left(\mathbf{b}\right)\right] + O(\varepsilon). \tag{$|\mathbf{Y}|^{\underline{k}} \leq |\mathbf{Y}|^{k}$}
\end{align}
We now relate the subgames $G_S$ to the restricted CSP objective on a random $\mathbf I$.
Sample $S_1,\ldots,S_k$ as in the theorem, and for each $i$ write $S_i=(y_{i,1},\ldots,y_{i,N})$.  For $j\in[N]$ define $\mathbf z_j := (y_{1,j},\ldots,y_{k,j})\in \mathbf Y$ and set $\mathbf I:=(\mathbf z_1,\ldots,\mathbf z_N)\in \mathbf Y^{N}$.
The tuples $\mathbf{z}_1, \dots, \mathbf{z}_N$ are all distinct except with probability at most $N^2/|\mathbf{Y}|$. Hence, $\mathbf{I}$ is uniformly distributed within $N^2/|\mathbf{Y}|$ total variation distance.

Fix one such $S$ with distinct $\mathbf{z}_j$ and consider any strategy $(b_i:S_i\to B_i)_{i\in[k]}$ for the subgame $G_S$.\footnote{A valid strategy must assign the same value to the same elements in $S_i$. In a slight abuse of notation we treat $b_i:S_i\to B_i$ as mapping from a set (without repetitions) to another set.}
Define an assignment $\mathbf b:\mathbf I\to \mathbf B$ by
\[
\mathbf b(\mathbf z_j):=\bigl(b_1(y_{1,j}),\ldots,b_k(y_{k,j})\bigr)
\]
for each $j\in[N]$.
Let $\mathbf R=(\mathbf z_{j_1},\ldots,\mathbf z_{j_k})$ be uniform in $\mathbf I^{\underline{k}}$ and let $\sigma\in\mathrm{Sym}_k$ be uniform.
Set $q_i:=(\mathbf z_{j_{\sigma(i)}})_i=y_{i,j_{\sigma(i)}}\in S_i$.
Since $(j_1,\ldots,j_k)$ is a uniformly random ordered $k$-tuple of distinct indices in $[N]$,
the joint distribution of $(q_1,\ldots,q_k)$ is within total variation distance at most $k^2/N$
of the product distribution $\bigotimes_{i=1}^k \mathrm{Unif}(S_i)$.
Because $V\in[0,1]$, it follows that
\[
\omega(G_S) \le \max_{\mathbf b:\mathbf I\to \mathbf B} \E_{\mathbf R\sim \mathbf I^{\underline{k}}}\bigl[C_{\mathbf R}(\mathbf b)\bigr]
+\frac{k^2}{N}.
\]
Taking the expectations gives
\[
\E_S[\omega(G_S)]
\le
\E_{\mathbf I \sim \mathrm{Unif}\big(\mathbf Y^{\underline{N}}\big)}\left[\max_{\mathbf b:\mathbf I\to \mathbf B}\E_{\mathbf R\sim \mathbf I^{\underline{k}}}[C_{\mathbf R}(\mathbf b)]\right]+\frac{k^2}{N} + \frac{N^2}{|\mathbf{Y}|}.
\]
Combining with the previous bounds yields the desired statement, as both $\frac{k^2}{N}$ and $\frac{N^2+k^2}{|\mathbf{Y}|}$ are $O(\varepsilon)$.
\end{proof}

\subsection{Property Testing}
In this section, we present a simple application of our subsampling theorem to property testing for satisfiability.
A natural class of testers in this setting is that of \emph{canonical} testers, which accept if and only if the subinstance induced by the sampled variables is satisfiable \citep{goldreich2003three}. Our tester is not strictly canonical, since our subsampling guarantee holds only in expectation rather than with high probability. Nevertheless, our tester is \emph{almost} canonical: we repeat the test independently $O(1/\varepsilon)$ times and reject if at least one instance is unsatisfiable.

Our subsampling theorem (\Cref{thm:subsampling}) directly gives a one-sided error tester for satisfiability with sample complexity $\poly\big(\frac{k\log q}{\varepsilon}\big)$. Intuitively, we have to repeat $O(1/\varepsilon)$ times the sampling in order to have a constant probability of rejecting instances that are $\varepsilon$-far from satisfiability. The proof is a simple application of Markov's inequality.\footnote{Since this theorem is a corollary of \Cref{thm:subsampling}, we have to assume that $n\ge \left(\frac{k\log q}{\varepsilon}\right)^c$ for a suitable constant $c$, which is, nonetheless, the interesting regime.}

\begin{theorem}\label{th:proptestingSAT}
    There exists a one-sided error $\varepsilon$-tester for $(k,q)$-\Sat with sample complexity $\tilde O(\frac{k^{12}}{\varepsilon^7}\log^2  (q))$.
\end{theorem}
\begin{proof}
    Given a boolean-valued CSP, we apply our subsampling theorem as follows. Consider the instance in which we set the value of every missing constraint to $1$ independently of the assignment.
    Define $\overline{\OPT}$ as $\OPT\cdot n^k/n^{\underline{k}}$, and similarly $\overline{\OPT}(U)=\OPT(U)\cdot N^k/N^{\underline{k}}$ for any $U\subseteq [n]$ of size $N$ (defined according to \Cref{thm:subsampling}).
    
    It is easy to see that the instance is satisfiable if and only if $\overline\OPT=1$, while if it is $\varepsilon$-far from satisfiable then $\overline{\OPT}\le 1-\varepsilon$.
    Let $\alpha_m=\frac{m^{\underline{k}}}{m^k}$.
   Using \Cref{thm:subsampling}, we get
    \begin{align}
        \mathbb{E}[\overline{\OPT}(U)]%
        &\le \frac{1}{\alpha_N}(\alpha_n \overline\OPT+O(\varepsilon))\tag{\Cref{thm:subsampling}}\\
        &\le \left(1+O\left(\frac{k^2}{N}\right)\right)(\overline\OPT+O(\varepsilon))\tag{\Cref{lem:trivialtmp}}\\
        &=\overline{\OPT}+C\varepsilon,\tag{Since $k^2/N=O(\varepsilon)$}
    \end{align}
    for some constant $C>0$. Note that by adjusting $N$, we can easily get 
    \begin{align}
        \mathbb{E}[\overline{\OPT}(U)]\le \overline{\OPT}+\frac\varepsilon2.\label{eq:tmp9876}
    \end{align}

    Now, we design our tester. The tester samples $\eta=O(1/\varepsilon)$ subinstances $U_1,\ldots, U_\eta$ of size $|U_j|=N$, and rejects if any of the subinstances $U_j$ is unsatisfiable. The total sample complexity is therefore  $N\eta=\tilde O\left(\frac{k^{12}}{\varepsilon^7}\log^2(q)\right)$.

    If the instance is satisfiable, then $\overline{\OPT}=1$ and then also $\overline{\OPT}(U_j)=1$ for all $j\in[\eta]$, and thus we will accept with probability $1$.

    On the other hand, if the instance is $\varepsilon$-far from being satisfiable, then for some absolute constant $C>0$ we get from \Cref{eq:tmp9876} that the probability that a random subinstance on $U_j$ is unsatisfiable can be lower bounded by Markov's inequality as:
    \[
    \mathbb{P}[1>\overline{\OPT}(U_j)]\ge 1-\mathbb{E}[\overline{\OPT}(U_j)]\ge \frac\varepsilon2,
    \]
    Thus, each sample yields an unsatisfiable subinstance with probability at least $\varepsilon/2$. It is clear that by our choice of $\eta$, the probability of accepting an instance which is $\varepsilon$-far from being satisfiable is less than $1/3$.
\end{proof}

It is also worth noting that our algorithm also provides a fully tolerant test \citep{parnas2006tolerant} for \Sat.
Namely, for any $\varepsilon_1<\varepsilon_2$, the tester accepts instances that are $\varepsilon_1$-close to satisfiable and rejects instances that are $\varepsilon_2$-far from satisfiable, with constant probability in both cases.
Its sample complexity has the same dependence on the gap $(\varepsilon_2-\varepsilon_1)^{-1}$ as our distance-estimation guarantee.
This follows because the algorithm does more than distinguish satisfiable instances from those far from satisfiability: it provides an additive approximation to the input instance's distance from satisfiability. Comparing this estimate with any threshold between $\varepsilon_1$ and $\varepsilon_2$ immediately gives a tolerant tester.

\section{Conclusions}

We proved a subsampling theorem for CSPs with sample size $\poly(\frac{k\log q}{\varepsilon})$. This exponential improvement on the dependency on $k$ compared to the result of \citet{barak2011subsampling} provides the missing ingredient to establish the collapse proven by \citet{aaronson2014multiple} that $\AM(\poly)=\AM$. We also applied our theorem to obtain the first testers for satisfiability with query complexity $\poly\big(\frac{k\log q}{\varepsilon}\big)$, exponentially improving the best known prior dependence on $q$ of \citet{blais2024new}, while maintaining polynomial dependence on both $k$ and $1/\varepsilon$.

We do not expect our dependence on $k$ and $\varepsilon$ to be optimal. In this work, we did not attempt to optimize exponents, and instead focused only on obtaining a $\poly\big(\frac{k\log q}{\varepsilon}\big)$ bound. Slightly better results could probably be gained with a more careful analysis. However, our techniques are unlikely to yield the optimal dependence on $\varepsilon$: for instance, when $q,k=O(1)$, the correct dependence is known to be $\Omega(1/\varepsilon)$ \citep{alon2002testing}. Tightening the dependence on $\varepsilon$ and obtaining a tester with sample complexity $O(\poly(k\log q)/\varepsilon)$ remains an interesting open challenge.

\appendix

\section*{Appendix}
\section{Useful Claims}

\begin{restatable}{claim}{claimSmallDifference}\label{lem:trivialtmp} For integers $1\le k\le N$ it holds that
    \(
    1-\frac{k^2}{N}\le\frac{N^{\underline{k}}}{N^k}\le 1.
    \)
\end{restatable}
\begin{proof}
    Note that $\frac{N^{\underline{k}}}{N^k}=\prod_{j=0}^{k-1}\left(1-\frac{j}{N}\right)$. Now, for all sequence of $(a_j)_{j \in [k]} \in [0,1]^k$, by induction we get that
    \[
    1-\prod_{j=0}^{k-1}(1-a_j)\le \sum_{j=0}^{k-1} a_j  
    \]
    Then, we have that for $a_j=j/N$,
    \[
    0\le1-\frac{N^{\underline{k}}}{N^k}\le \sum_{j=0}^{k-1}\frac jN = \frac{(k)(k-1)}{2N}\le \frac{k^2}{N},
    \]
    concluding the proof.
\end{proof}

\begin{restatable}{claim}{withToWithout}\label{clm:withToWithout}
     Let $U_1$ be a set of $N$ elements sampled uniformly from $[n]$ without replacement, 
     and let $U_2$ be the set of variables resulting from sampling $N$ elements with replacement uniformly from $[n]$. 

     Consider the event $E$ that $U_2$ has no collisions, i.e., $E=\{|U_2|=N\}$. Then
     \begin{itemize}
         \item $\mathbb{P}(E)\ge 1-\frac{N^2}{n}$;
         \item $U_2$ conditioned on $E$ has the same law of $U_1$.
     \end{itemize}
\end{restatable}

\begin{proof}
Given the $N$ i.i.d.~$I_t\in[n]$, let $U$ be the set of sampled variables, i.e., 
\[
U_2:=\{i\in[n]:\exists t\in[N]\text{ s.t. } I_t=i\}.
\]

We exploit the following result: the probability that we sample twice any variable is
\begin{align}\label{lem:collision}
\mathbb{P}(\exists s<t\in[N]: I_s=I_t)\le \mathbb{P}(\exists s\neq t\in[N]: I_s=I_t)\le \binom{N}{2}\frac1n\le \frac{N^2}n.
\end{align}

Let $E$ be the event that there are no collisions, i.e.\ all $I_1,\dots, I_N$ are distinct.
Then, under event $E$, we have $|U_2|=N$. Moreover, it is easy to check that the conditional law of $U_2$ given $E$ is uniform over all $N$-subsets of $[n]$. Thus, $U_1$ and $U_2$ have the same distribution conditioned on $E$.

By \Cref{lem:collision}, we can lowerbound the probability of event $E$ as:
\begin{equation*}
\mathbb{P}(E)\ge 1-\frac{N^2}{n},
\end{equation*}
concluding the proof.
\end{proof}

\section{Proofs Omitted from \Cref{sec:structure}}\label{app:omitted}

\lemmaDecomposition*
\begin{proof}

    The proof relies on the fact that $\cM(U)$ can be partitioned into all constraints that have all variables inside $U_{-t}$ (which is the set $\cM(U_{-t})$) or one variable in $U_t$ and all others in $U_{-t}$ (which is the set $\bigcup_{i\in U_t}\cM_i(U)$ ). This follows because, in $\cM$, we consider only balanced constraints. In particular, we have
    \begin{align*}
        \val^\bal_U(\phi)&=\frac{1}{|U|^k}\sum_{w\in \cM(U)} P_w(\phi\mid_w)\\
        &=\frac{1}{|U|^k}\sum_{w\in \cM(U_{-t})} P_w(\phi\mid_w)+\frac{1}{|U|^k}\sum_{w\in \bigcup_{i\in U_t}\cM_i(U)} P_w(\phi\mid_w)\\
        &=\frac{1}{|U|^k}\sum_{w\in \cM(U_{-t})} P_w(\phi\mid_w)+\frac{1}{|U|^k}\sum_{i\in U_t}\sum_{w\in \cM_i(U)} P_w(\phi\mid_w)\\
        &=\frac{1}{|U|^k}\sum_{w\in \cM(U_{-t})} P_w(\phi\mid_w)+\frac{1}{|U|^k}\sum_{i\in U_t}\sum_{w\in \cM_i(\{i\}\cup U_{-t})} P_w(\phi\mid_w).
    \end{align*}
    For each $i\in U_t$, the inner summation is precisely $H_{i,U_{-t}}(\phi_i,\phi)$, proving the decomposition.
\end{proof}

\equivalencelaws*

\begin{proof}

Consider the joint distribution $(U, \ccS) \sim \gamma$. By the construction of the coupling, the conditional probability of observing a specific seed $\ccS$ given $U$ factorizes over the blocks:
\[
\mathbb{P}(\ccS \mid U) = \prod_{t \in [T]} \mathbb{P}(S_t \mid U) = \prod_{t \in [T]} \mathbb{I}(S_t \subseteq U_{-t}) \binom{|U_{-t}|}{s}^{-1}.
\]
In our sampling procedure, the sizes $|U_t|$ are deterministic constants. Therefore, 
\[\mathbb{P}(\ccS \mid U) \propto \mathbb{I}(S_t \subseteq U_{-t} \, \forall t \in [T]).\]

Now, it is easy to see that, since each $S_t$ is drawn from $V_{-t}$ by definition $S_t\cap V_t=\emptyset$, and thus
\[\mathbb{I}(S_t \subseteq U_{-t} \, \forall t \in [T])= \mathbb{I}(S\subseteq U).\]
for each $(S_t)_{t \in T}$, $U$ in the support of $\gamma$.

By Bayes' Theorem, the conditional distribution of $U$ given $\ccS$ is:
\[
\mathbb{P}(U \mid \ccS) \propto \mathbb{P}(U) \mathbb{P}(\ccS \mid U) \propto \mathbb{P}(U) \mathbb{I}(S \subseteq U).
\]
This is precisely the law $\nu_S$.
\end{proof}

\printbibliography

\end{document}